\documentclass%
[a4paper,USenglish,
cleveref,autoref,thm-restate]%
{lipics-v2021}

\pdfoutput=1 

\hideLIPIcs 

\usepackage{xspace}
\usepackage{tikz}
\usepackage{svg} 
\usepackage{upgreek} 

\usepackage{enumitem} 
\setlist[itemize]{%
  label=\textbullet, 
  parsep=1ex} 

\long\def\Cyril#1{}
\long\def\Arnaud#1{}
\long\def\NicoB#1{}
\long\def\NicoH#1{}

\graphicspath{{./figures/}} 
\nolinenumbers

\def\eta{\upeta} 

\def\lme{\Lambda(\eta)} 
\def\elli{\ell_{\infty}}
\def\gme{\gamma(\eta)}

\newtheorem{fact}{Fact}
\def\myparagraph#1{\medskip\noindent\textbf{#1}}

\newcommand{\FTP}{\textsc{Freeze-Tag Problem}\xspace}
\newcommand{\eps}{\varepsilon}
\newcommand{\sfP}{\mathsf{P}}
\newcommand{\sfNP}{\mathsf{NP}}

\newcommand{\bbR}{\mathbb{R}}
\newcommand{\bbN}{\mathbb{N}}
\newcommand{\set}[1]{\left\{{#1}\right\}}
\newcommand{\range}[2]{\set{{#1},\dots,{#2}}}
\newcommand{\pare}[1]{\left({#1}\right)}
\newcommand{\floor}[1]{\left\lfloor{#1}\right\rfloor}
\newcommand{\ceil}[1]{\left\lceil{#1}\right\rceil}

\def\heapstr{\textup{\texttt{Heap-Strategy}}\xspace}
\def\splitconestr{\textup{\texttt{Split-Cone-Strategy}}\xspace}
\def\linearstr{\textup{\texttt{Linear-Split-Strategy}}\xspace}
\def\greedystr{\textup{\texttt{Greedy-Strategy}}\xspace}

\def\arc{\mathsf{arc}}
\def\cone{\mathsf{cone}}

\let\oldthebibliography\thebibliography
\renewcommand\thebibliography[1]{%
  \oldthebibliography{#1}%
  \setlength{\labelsep}{2mm}%
  \setlength{\labelwidth}{8mm}%
}

\title{Freeze-Tag in $L_1$ has Wake-up Time Five}

\titlerunning{Freeze-Tag in \boldmath{$L_1$} has Wake-up Time Five} 

\author{Nicolas Bonichon}%
{LaBRI, University of Bordeaux, CNRS, Bordeaux INP, France}%
{bonichon@labri.fr}%
{}{}

\author{Arnaud Casteigts}%
{LaBRI, University of Bordeaux, CNRS, Bordeaux INP, France \and CS Department, University of Geneva, Switzerland}%
{acasteig@labri.fr}%
{}{}

\author{Cyril Gavoille}%
{LaBRI, University of Bordeaux, CNRS, Bordeaux INP, France}%
{gavoille@labri.fr}%
{https://orcid.org/0000-0003-3671-8607}{}

\author{Nicolas Hanusse}%
{LaBRI, University of Bordeaux, CNRS, Bordeaux INP, France}%
{hanusse@labri.fr}%
{}{}

\Copyright{Nicolas Bonichon, Arnaud Casteigts, Cyril Gavoille, and
  Nicolas Hanusse}

\authorrunning{N. Bonichon, A. Casteigts, C. Gavoille, N. Hanusse}



\begin{document}
\maketitle

\Cyril{$\ell_p$-norm = c'est le nom de la norme, par exemple
  l'occurrence $\ell_p$, alors que $L_p$ space = c'est le nom d'un
  espace muni de la norme $\ell_p$. Par exemple, "dans $L1$" veut dire
  "dans $(R^2,\ell_1)$". Par contre, "dans $\ell_1$" ne veut rien dire
  a priori.}

\Cyril{Je trouve que wake-up time est un poils plus clair que
  makespan. Ceci dit, j'aime bien makespan aussi. Dans la mesure où
  l'on parle de wake-up tree, cela devient lourding de dire par
  exemple: "There is a wake-up tree of wake-up time ... ". Bon, on
  parle de wake-up time dans titre, mais c'est ok, les deux sont en
  fait synonymes. C'est juste qu'il faut être cohérent dans le
  statements: toujours makespan ou toujours wake-up time, sauf le
  titre :)). Attention aussi quand on dit que le robot X réveille Y
  "in time $t$". C'est pas très clair: on pourrait confondre avec le
  temps $O(n)$ pour construire le wake-up tree ... On devrait toujours
  parler de makespan, de wake-up time ou plus simplement ici de
  longueur d'arête.}

\Cyril{On ne met pas de refs [...] sous la forme de \cite{} dans un
  abstract, car sinon l'abstract n'est plus self-contains ce qui est
  mauvais. On peut mettre cependant une forme développée: [title -
  Proceeding of ... - 1984]. Eviter aussi les macros dans l'abstract,
  source d'erreur lorsqu'on fait les copier/coller.}

\begin{abstract}
  The \textsc{Freeze-Tag Problem}, introduced in Arkin et
  al. (SODA'02) consists of waking up a swarm of $n$ robots, starting
  from a single active robot. In the basic geometric version, every
  robot is given coordinates in the plane. As soon as a robot is
  awakened, it can move towards inactive robots to wake them up. The
  goal is to minimize the wake-up time of the last robot, the
  \emph{makespan}.

  Despite significant progress on the computational complexity of this
  problem and on approximation algorithms, the characterization of
  exact bounds on the makespan remains one of the main open
  questions. In this paper, we settle this question for the
  $\ell_1$-norm, showing that a makespan of at most $5r$ can always be
  achieved, where $r$ is the maximum distance between the initial
  active robot and any sleeping robot.
  Moreover, a schedule achieving a makespan of at most $5r$ can be
  computed in optimal time $O(n)$. Both bounds, the time and the
  makespan are optimal. This implies a new upper bound of
  $5\sqrt{2}r \approx 7.07r$ on the makespan in the $\ell_2$-norm,
  improving the best known bound so far
  $(5+2\sqrt{2}+\sqrt{5})r \approx 10.06r$.

\end{abstract}


\section{Introduction}
\label{sec:intro}

The \FTP (FTP) is an optimization problem that consists of activating
as fast as possible a swarm of robots represented by points in some
metric space (in general, not necessarily euclidean). Active (or awake) robots can move towards any point of
the space at a constant speed, whereas inactive robots are asleep (or
frozen) and can be activated only by a robot moving to
their positions. Initially, there are $n$ sleeping robots and one awake
robot. The goal is to determine a schedule whose makespan is
minimized; that is, the time until all the robots have been activated
is minimized. FTP has application not only in \emph{robotics},
e.g. with group formation, searching, and recruitment, but also in
\emph{network design}, e.g. with broadcast and IP multicast
problems. See~\cite{ABGHM03,ABFMS06,KLS05} and references therein.


FTP is $\sfNP$-Hard in high dimension metrics like centroid
metrics~\cite{ABFMS06} (based on weighted star $n$-vertex graphs) or
unweighted graph metrics with a robot per node~\cite{ABGHM03}. Many
subsequent works have extended this hardness result to constant
dimensional metric spaces, including the Euclidian ones. A serie of
papers~\cite{AAY17,Johnson17,PdOS23} proves that FTP is actually
$\sfNP$-Hard in $(\bbR^3,\ell_p)$, for every $p\ge 1$, i.e., in 3D
with any $\ell_p$-norm. For 2D spaces, this remains $\sfNP$-Hard
for $(\bbR^2,\ell_2)$, leaving open the question for other
norms~\cite{AAY17}. It is believed,
see~\cite[Conjecture~28]{ABFMS06}, that FTP remains $\sfNP$-Hard
for $(\bbR^2,\ell_1)$.


Several approximation algorithms and heuristics have
been designed. In their seminal work, \cite{ABFMS06} developed a
$14$-approximation for centroid metrics, and a PTAS for
$(\bbR^d,\ell_p)$. \cite{ABGHM03} presented a $O(1)$-approximation
for unweighted graph metric with one robot per node, and a greedy
strategy analyzed in~\cite{SABM04} gives a
$O(\log^{1-1/d}{n})$-approximation in $(\bbR^d,\ell_p)$. For general
metrics, the best approximation ratio is
$O(\sqrt{\log{n}}\,)$~\cite{KLS05}.  For heuristics, several
experiments results can be found in
\cite{Bucatanschi04,BHHK07,Keshavarz16}.
See~\cite{AAS10,MB14,HNP06,BW20} for generalizations and variants of
the problem, including the important online version.


As observed by~\cite{ABFMS02}, the FTP can be rephrased as finding a
rooted spanning tree on a set of points with minimum weighted-depth,
where the root node (corresponding to the awake robot) has one child
and all the others nodes (corresponding to the $n$ sleeping robots) have
at most two children (see Figure~\ref{fig:exL2}). Each edge has a \emph{length}, a non-negative
real, representing the distance in the metric space between its
endpoints. Such a tree is called a \emph{wake-up tree}, and its
weighted-depth is called \emph{makespan}.

\begin{figure}[htbp!]
  \centering%
  \includesvg[width=0.8\textwidth]{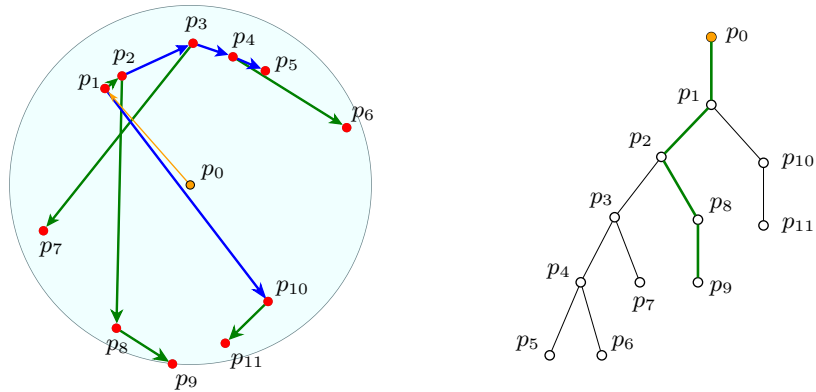}%
  \caption{Example of a (here, euclidean) instance of FTP (on the left). The robot at
    $p_0$ must wake up $n = 11$ sleeping robots at $p_1,\dots,p_n$. In
    this example, positions are normalized in the unit $\ell_2$-disk,
    $p_0$ being at the center. The optimal solution, depicted by
    arrows, can be represented as a binary weighted tree (right).
    The makespan is the length of the longest (weighted) branch in that tree, here $2.594$,
    corresponding to the path $(p_0,p_1,p_2,p_8,p_9)$.
    Observe that, even if the sleeping robots are in a convex
    configuration, the optimal solution may have multiple edge
    crossings. }
  \label{fig:exL2}
\end{figure}


Clearly, FTP is related to the \textsc{Traveling Salesperson Problem},
in its \emph{metric} version (hereafter, simply TSP).  Indeed, the
Path-TSP, a generalization of TSP in which one ask for finding a
minimum path length spanning a point set from given start and end
points, provides a valid wake-up tree and thus a solution for FTP. The
link with TSP is reenforced by the recent approximated reduction of
Path-TSP to classical TSP~\cite{TVZ20}: if there is an
$\alpha$-approximation for TSP, then, for every $\eps>0$, there is a
$(\alpha+\eps)$-approximation for Path-TSP. Recently~\cite{KKO21}
showed that there exists $\alpha < 3/2 - 10^{-36}$ for TSP, a
qualitative breakthrough since the 1976 Christofides-Serdyukov
algorithm.

This being said, there are significant differences between TSP and FTP, the
latter being considered as a cooperative TSP version where awaked
robots can help in visiting unvisited cities. First, from an
algorithmic point of view, the best lower bound on the approximation
factor are $5/3 - \eps$ for FTP~\cite{ABFMS06}, and only
$123/122 - \eps$ for TSP~\cite{Karpinski15} (assuming
$\sfP \neq \sfNP$). On the other side, the time complexity for PTAS in
$(\bbR^d,\ell_p)$ is $n(\log{n})^{(d/\eps)^{O(d)}}$ for
TSP~\cite{Arora98,RS98,Mitchell99} vs.
$O(n\log{n}) + 2^{(d/\eps)^{O(d)}}$ for FTP~\cite{ABFMS06}, subject to
$\eps \le \eps_d$, where $\eps_d$ depends on the
number of dimensions~$d$. Second, and perhaps more fundamentally, it is well
known that, even in the unit ball in
$(\bbR^d,\ell_p)$, the shortest spanning path (or tour) has unbounded
length in the worst-case (it depends on $n$), whereas the makespan 
for FTP is bounded by an absolute constant (that does not depend on $n$). For TSP, the
worst-case length is $\Theta_d(n^{1-1/d})$~\cite{Few55}, whereas for
FTP the worst-case optimal makespan is no more than some
constant $\rho_d$, so independent of $n$~\cite{ABFMS06}.

The constant $\rho_d$ plays an important role for PTAS and
approximation algorithms. For instance, it drives the condition
``$\eps \le \eps_d$'' in the grid refinement approach
of~\cite{ABFMS06}, where local solutions in radius-$(1/\eps)$ balls
have to be constructed. For $(\bbR^2,\ell_2)$, the constant $\rho_2$
coming from the approach of~\cite{ABFMS06} has been proved to be at
most~$57$ by~\cite{YBMK15}. The latter authors have also constructed
in time $O(n)$ a wake-up tree of makespan at most
$5 + 2\sqrt{2} + \sqrt{5} \approx 10.06$, which is the best known
upper bound for~$\rho_2$.

\subsection*{Our contributions}



In this paper, we concentrate our attention to the plane $\bbR^2$.
Given a norm $\eta$, the \emph{unit disk w.r.t. $\eta$}, or the unit
$\eta$-disk for short, is the normed linear subspace of
$(\bbR^2,\eta)$ induced by all the points at distance at most one from
the origin, where distances are measured according to $\eta$, the
distance between $u$ and $v$ being $\eta(v-u)$. The unit $\eta$-disk
can be an arbitrarily convex body that is symmetric about the origin.
Note that the unit $\ell_2$-disk\footnote{For convenience, and to
  avoid extra notation, we use the same ``unit $\eta$-disk''
  terminology to denote the normed subspace and, like here, its
  support, that is the set of all points of norm at most~$1$ (the
  disk).}
is an usual disk 
whereas the unit $\ell_1$-disk is a rotated $45$ degrees
square. 

Our main contribution is the following. 

\begin{restatable}{theorem}{TheoremMain}\label{th:main}
  A robot at the origin can wake up any set of $n$ sleeping robots
  in the unit $\ell_1$-disk with a makespan of at most~$5$. The
  wake-up tree can be constructed in $O(n)$ time.
\end{restatable}

Obviously, if the awake robot is at distance at most $r$ from all the
sleeping robots, then by scaling the unit disk with their positions, and
by using \cref{th:main}, one can construct a wake-up tree of makespan
of at most $5r$. By a loose argument, this yields a $5$-approximation
$O(n)$ time algorithm for $(\bbR^2,\ell_1)$, since $r$ is a trivial
lower bound on the makespan. As we will see, a similar statement holds
for $(\bbR^2,\elli)$.

Both bounds in \cref{th:main} are optimal: the makespan of~$5$, and
obviously the linear time construction of the wake-up tree. The upper
bound of~$5$ is reached with $n = 4$ sleeping robots at positions
$(\pm 1,0)$ and $(0,\pm 1)$. Indeed, any wake-up tree spanning more
than four points must have (unweighted) height at least~$3$. Then, the
first hop has length~$1$ and the next two hops have
length~$2$, which overall gives a makespan of at least~$5$ for any
wake-up tree. Actually, we will see in \cref{th:gamma4} a
generalization of this argument for any norm $\eta$, leading to an
intriguing open question of matching this lower bound for other norms
(see \cref{conj}).

By a simple argument, \cref{th:main} immediately improves the best
known upper bound for norm $\ell_2$. Indeed (see also \cref{cor:ellp}), by
scaling the $\ell_2$-disk, we can use the construction of
\cref{th:main} to obtain a makespan of $5\sqrt{2} \approx 7.07$ for
the unit $\ell_2$-disk, improving upon the previous $10.06$ upper
bound of~\cite{YBMK15}.





Our second result concerns algorithmic aspects of the FTP.
\cref{th:time} states that there is always a linear time algorithm
that can match the best upper bound for wakening a unit disk. The
result is general enough to hold in any normed linear space
$(\bbR^2,\eta)$, akka Minkowski plane.

To make the statement of \cref{th:time} precise, let us define
$\gamma_n(\eta)$ as the worst-case optimal makespan of a wake-up
tree for any set of $n$ sleeping robots in the unit $\eta$-disk and
rooted at the origin. In other words, $\gamma_n(\eta)$ is the best
possible upper bound of the makespan to wake-up $n$ sleeping robots from
an awake robot placed at the origin, in the unit $\eta$-disk. Finally,
let us introduce the \emph{wake-up ratio} w.r.t. the $\eta$-norm
defined by
\[
  \gme ~=~ \max_{n\in\bbN} \gamma_n(\eta) ~.
\]
Note that the constant $\rho_2$ introduced above is nothing else than
$\gamma(\ell_2)$, the $\ell_2$ wake-up ratio.

\begin{restatable}{theorem}{TheoremLinearTime}\label{th:time}
  Let $\eta$ be any norm and let $\tau>3$ be any real such that
  $\tau\ge \gme$. Knowing $\tau$, one can construct in time $O(n)$ a wake-up
  tree of makespan at most $\gme$ for any set of $n$ points in
  the unit $\eta$-disk and rooted at the origin.
\end{restatable}

So, plugging $\eta = \ell_1$ and $\tau = 5$ in \cref{th:time}, it is
sufficient to prove that $\gamma(\ell_1) \le 5$ to automatically
obtain a linear time construction of a wake-up tree of makespan at
most~$5$ as claimed in \cref{th:main}. In other words, given
\cref{th:time}, our main \cref{th:main} could simply be restated as:
$\gamma(\ell_1) \le 5$. Furthermore, as already explained, the bound
of~$5$ is attained for $n = 4$ sleeping robots, so
$\gamma(\ell_1) \ge \gamma_4(\ell_1) = 5$, and the
wake-up ratio in $\ell_1$-norm is thus $5$.

To prove \cref{th:main} and \cref{th:time}, we need several
intermediate results, which we believe are of independent interest. For
instance, we show how to efficiently wake up robots contained
in a unit cone of given arc-length. This implies, for instance, that
$\gme < 3 + 4\varphi = 5 + 2\sqrt{5} \approx 9.47$, where $\varphi$ is
the golden ratio. 

Some of our arguments rely on an intermediate result that we prove for
all possible norms. Given a norm $\eta$, let us define $\lme$ as half
the perimeter of the largest inscribed parallelogram in the unit
$\eta$-disk (in $\ell_1$-norm, this perimeter is the disk itself). This is a classical parameter for normed spaces. It can
be formally defined by (see~\cite{Schaffer76,Gao01}),
\[
  \lme ~=~ \sup\limits_{\substack{u,v
      \in\bbR^2\\\eta(u),\eta(v)\le 1}} \set{~ \eta(u+v) + \eta(u-v)
    ~} ~.
\]
It is easy to check that $\lme \in [2,4]$. For general norms, the
quantity $\lme$ is difficult to calculate. However, it is known
(see~\cite[Proposition~1]{Gao01} for instance), that, for every
$p \in [1,\infty]$, $\Lambda(\ell_p) = 2^{1+\max(1/p,1-1/p)}$. In
particular, $\Lambda(\ell_1) = \Lambda(\elli) = 4$ and
$\Lambda(\ell_2) = 2\sqrt{2}$.

The wake-up ratios for $n\in\set{0,1,2,3}$ are easy to calculate.
We have $\gamma_0(\eta) = 0$,
$\gamma_1(\eta) = 1$, $\gamma_2(\eta) = \gamma_3(\eta) = 3$, and also
$\gamma_n(\eta) \ge 3$ for all\footnote{For $n\ge 2$, it is enough to
  place one sleeping robot at $(1,0)$ and the $n-1$ others at $(-1,0)$.}
$n \ge 3$. Our next result gives the exact value for $\gamma_4(\eta)$.

\begin{restatable}{theorem}{TheoremGammaFour}\label{th:gamma4}
  For any norm $\eta$, $\gamma_4(\eta) = 1 + \lme$.
\end{restatable}

This implies a general lower bound of $\gme \ge 1+\lme$. Note that
since $\lme \ge 2$, \cref{th:time} simplifies and rewrites in:

\begin{corollary}\label{cor:time}
  For any norm $\eta$ with $\lme \neq 2$, one can construct in time
  $O(n)$ a wake-up tree of makespan at most $\gme$ for any set of $n$
  points in the unit $\eta$-disk and rooted at the origin.
\end{corollary}

Now, combining \cref{th:main}, \cref{th:gamma4} and standard inclusion
arguments of unit $\ell_p$-disk, we get the following bounds for the
$\ell_p$ wake-up ratio:

\begin{restatable}{corollary}{CorollaryEllp}\label{cor:ellp}
  For every $p \in [1,\infty]$,
  $1 + 2^{1+\max(1/p,1-1/p)} \le \gamma(\ell_p) \le 5 \cdot
  2^{\min(1/p,1-1/p)}$.
\end{restatable}

In the light of the lower bound $\gme \ge 1+\lme$ implied by
\cref{th:gamma4}, we propose the following natural conjecture.

\begin{restatable}{conjecture}{ConjectureMain}\label{conj}
  For any norm $\eta$, $\gme = 1 + \lme$.
\end{restatable}

According to \cref{th:gamma4}, which states that the bound $1 + \lme$
is reached by $n = 4$ robots, \cref{conj} can be captured in the
aphorism:

\begin{quote}\it
  It's always quicker to wake up $n$ robots than four.
\end{quote}

\cref{th:main} and \cref{cor:ellp} prove the conjecture for
$\eta \in\set{\ell_1,\elli}$. For $\eta = \ell_2$, if true, \cref{conj}
combined with \cref{th:time} imply that in time $O(n)$
one can construct a wake-up tree of makespan
$1 + 2\sqrt{2} \approx 3.82$. \cref{prop:gamma_n} implies also that
\cref{conj} is true for $\ell_2$ whenever $n \ge 528$.

\section{The wake-up ratio is at most 5 in $L_1$}
\label{sec:wake-up-time}

Our main result (\cref{th:main}) is to prove that the wake-up ratio
for the $\ell_1$-norm is at most~$5$. The proof is constructive and
provides a polynomial time algorithm. The complexity is subsequently
improved to a $O(n)$ in \cref{sec:linear}.

At a high level, the strategy consists of recruiting first a team of
robots in a dense subregion, then these robots can wake up the other
regions in parallel. The difficult part is to select these regions
appropriately, depending on the number of robots and their
distribution, and to prove that the bound holds in all the cases.

To make this more precise, let us partition the unit $\ell_1$-disk
into \emph{squares} and \emph{triangles} as follows.  A \emph{square}
(of diameter 1) is a square region whose diagonals and sides are both
of length $1$ ($\ell_1$ norm) and the diagonals are parallel to the
$x$-axis and $y$-axis (see~\cref{fig:partition}). Similarly, a
\emph{triangle} (of diameter~$1$) is an isosceles right triangle whose
hypothenuse and sides have length $1$ and the hypotenuse is parallel
to the $x$-axis or $y$-axis. Thus, each square region represents a
fourth of the unit $\ell_1$-disk, possibly subdivided further into two
equal triangles. Strictly speaking, the diameter may be smaller than
$1$, as our algorithm occasionally subdivides some regions further,
but the arguments are then normalized to $1$ systematically. Finally,
note that both squares and triangles can be seen as \emph{cones} (of
different angles) in the $\ell_1$-norm.

\Cyril{En Anglais c'est "isosceles right triangle" et pas "right
  isosceles triangle".}

\begin{figure}[htbp!]
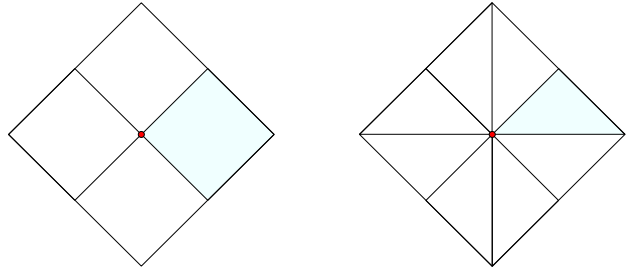

  \def\w{0.25}
  \centering%
  \includesvg[width=\w\textwidth]{partitionS.svg}%
  \hspace{1cm}
  \includesvg[width=\w\textwidth]{partitionT.svg}%
  \caption{The unit $\ell_1$-disk, divided into squares and triangles of
    diameter~$1$.}
  \label{fig:partition}
\end{figure}

The algorithm relies crucially on three
lemmas about these regions, namely:

\Cyril{J'aurai tendance à remplacer "time" par "makespan" dans les
  lemmes. En toute logique "time" se réfère au temps de construction.}

\begin{restatable}{lemma}{LemmaSfive}\label{lem:S5}
  A robot located at a corner of a square can wake up any number
  $n\le 5$ of robots in the square in two time units.
\end{restatable}

\Cyril{J'ai restreins \cref{lem:S6} à exactement $n = 6$ robots, car
  ce lemme n'est utilisé que dans ce cas dans \cref{th:main}. Bien sûr
  ce n'est pas faux de dire $n\le 6$, mais cela ne rajoute rien. En
  fait, si une ligne dans la preuve de \cref{lem:S6}.}

\begin{restatable}{lemma}{LemmaSsix}\label{lem:S6}
  A robot located at a corner of a square can wake up $6$ robots in
  the square and return to the origin with these robots in three time
  units.
\end{restatable}

\begin{restatable}{lemma}{LemmaT}\label{lem:T}
  A robot located at any of the three corners of a triangle $T$, or
  two robots located at a same point on a side of $T$ (not the
  hypotenuse) can wake up all the robots in $T$ in two time units.
\end{restatable}

A significant part of the paper is devoted to proving these lemmas. In
particular, \cref{lem:T} is proved by an induction involving 14
subcases. Equipped with these lemmas, the algorithm can be described
in a compact way as follows.

\Cyril{Il y'a une petite confusion quand on parle de "square of
  diameter $r$". On veut dire deux choses en une: "un carré de
  diamètre $r$" c'est un carré tel que ... Et si $S$ est un carré de
  diamètre $r$ selon la définition précédente, alors la plus grande
  des distances, son diamètre, se trouve être $r$. Je suggère donc
  d'utiliser dans la définition de carré "diagonal", puis de dire:
  "Observe that a square of diagonal $r$ has diameter $r$." Idem pour
  triangle.}

\Cyril{Pour le cas $6 \le n_0 \le 10$, je n'utilise pas le
  \cref{lem:T}. Cela me paraît plus économe, le cas dur étant
  concentrer pour la fin ...} \Arnaud{J'ai remplacé par une preuve ad
  hoc, très simple, pour éviter la dépandance à une autre section}
  
  \begin{proof}[The algorithm]
    The strategy is split into four scenarios as follows, depending on
    the number $n_0$ of robots in the densest square.

    \begin{itemize}


    \item \textbf{\boldmath{$n_0 = 1$.}} In this case, there are at
      most four robots to be awakened. The initiator wakes up one of
      them in one time unit. We now have two awake robots. Each of
      them independently wakes up another sleeping robot, in at most
      two time units (largest possible distance within the unit
      $\ell_1$-disk).  Then, any of the awake robot wakes up the last
      robot in at most two time units, which gives a total makespan of
      at most $1+2+2 = 5$.

  \item \textbf{\boldmath{$2\le n_0 \le 5$.}} We recruit $n_0$ robots
    from the densest square $S$ in two time units (\cref{lem:S5}),
    then come back to the origin (by time $2+1=3$) with
    $n_0 + 1 \ge 3$ awake robots. Since $S$ is the densest square,
    then three of the awake robots at the origin can each wake up one
    of the remaining squares (\cref{lem:S5}) in two time units, which
    gives a total of at most $3 + 2 = 5$.

  \item \textbf{\boldmath{$6 \le n_0 \le 10$.}} We recruit $6$ robots
    (chosen arbitrarily) in the densest square $S$ and move them to
    the origin in $3$ time units (\cref{lem:S6}). Together with the
    initiator, this makes $7$ robots. One of them wakes up the
    remaining robots in $S$, which are at most $4$, in two time units
    (\cref{lem:S5}). The~$6$ others split into three teams of two
    robots, one team for each remaining square, and each robot wakes
    up half of the sleeping robots in its assigned square, again in
    two time units (\cref{lem:S5}), which gives a total of at most $3 + 2 = 5$.

  \item \textbf{\boldmath{$n_0 \ge 11$.}} The densest square $S$ must
    contain a triangle $T$ with at least $\ceil{n_0/2} \ge 6$ sleeping
    robots. We wake up all the robots of $T$ in 2 time units
    (\cref{lem:T}) and move them to the origin. This makes
    at least $7$ robots. Each of them wakes up a remaining triangle
    in 2 time units (\cref{lem:T} again), which gives a
    total of at most $2+1+2 = 5$.

  \end{itemize}
  And this completes the proof of \cref{th:main}.
\end{proof}

\subsection{Preamble}

\Cyril{J'aurai tendance à dire "... robots are represented by a
  sequence $p_0, p_1, \dots, p_n$ of points taken ...", l'idée étant
  qu'on peut avoir plusieurs robots à la même position, $p_i = p_j$
  par exemple. Cela ne gêne en rien la preuve et nos résultats.}
\Arnaud{J'ai mis ``multiset'' à la place.}

The positions of robots are given as a multiset of points
$\{p_0, p_1, \dots, p_n\}$ taken in the unit $\ell_1$-disk, where
$p_0$ is the position of the awake robot, and $p_1,\dots,p_n$ the
positions of the $n$ sleeping robots. We use the notation $|p_ip_j|$
to denote the $\ell_1$-distance between $p_i$ and $p_j$, i.e.,
$|p_ip_j| = \ell_1(p_j-p_i) = |x(p_j)-x(p_i)| + |y(p_j)-y(p_i)|$.

\Cyril{Ne faut-il pas le sens large pour la relation d'ordre?
  $p <_x q$? Car même si $p\neq q$, on pourrait avoir $x(p) = x(q)$.}

We define two orderings for points $p,q \in \bbR^2$. Namely,
$p \le_x q$ if the $x$-coordinate of $p$ is no more than the one of
$q$.  We also say the $p$ is on \emph{the right} of $q$ (or that $q$
is on \emph{the left} of $p$). Similarly, $p \le_y q$ if the
$y$-coordinate of $p$ is no more than the one of $q$, and we say that
$p$ is \emph{below} $q$ (or that $q$ is \emph{above} $p$).


\Cyril{J'ai utilisé la notation $A-B-C-...$ pour les chemins,
  notamment dans les preuves de \cref{th:gamma4}. Pour les formes
  géométriques, disons un triangle, j'utilise $ABC$ en collant les
  lettres, plutôt que $(A,B,C)$ ou autre variante $A,B,C$ qu'on peut
  confondre avec les chemins. Il faut essayer de s'y tenir. }
\Arnaud{(a,b,c) pour les chemins me semble plus standard, mais à voir...
je tente ça.}

\paragraph*{Monotonic paths}
A path $(p_0, p_1, \cdots, p_t)$, $t\ge 1$ is \emph{monotonic} if it
is compatible with both $\le_x$ and $\le_y$. More precisely,
for each $i \in\range{1}{t}$, we must have
($p_{i-1}\le_x p_i \Leftrightarrow p_0\le_x p_1$) and
($p_{i-1}\le_y p_i \Leftrightarrow p_0\le_y p_1$). In other words,
monotonic paths use points that are always going in the same direction
w.r.t. the four quadrants: North-East (NE), North-West (NW),
South-West (SW), South-East (SE).

\Cyril{Il va falloir se poser la question des robots à la même
  position. D'un coté on s'en fout: Qui peut le plus peut le
  moins. Réveiller deux robots d'un coup, pourquoi. En plus, il arrive
  que les robots se déplacent tous vers l'origine par exemple. Donc
  avoir plusieurs robots au même endroit fait sens. D'un autre coté,
  l'input devrait alors être une séquence et pas un ensemble. Cela
  pose aussi problème avec $p <_x q$ où clairement il faudrait que
  $p\neq q$. J'aurai tendance à dire que l'input $p_0,p_1,...,p_n$ est
  bien un ensemble deux à deux différents de positions et donc de
  robots. Que les robots se déplacent aux mêmes endroits et un autre
  problème. Il n'empêche que le wake-up tree doit couvrir le $p_i$,
  peut importe les chemins.}

A fundamental property of the $\ell_1$-norm is that all monotonic paths
are shortest paths. More formally, if $(p_0, p_1, \ldots, p_t)$ is a monotonic path, then
\[
  \sum_{1 \le i \le t} |p_{i-1} p_i| ~=~ |p_0 p_t| ~.
\]

A path is \emph{$k$-monotonic} if it can be subdivided into
$k$ monotonic subpaths. We thus have:

\begin{lemma}\label{lem:monotonic}
  In a region of diameter $D$, the length of a $k$-monotonic path is at
  most $k D$.
\end{lemma}

Our algorithm exploits monotonic paths on several occasions, in order to wake up intermediate robots at no additional costs. 

\begin{lemma}\label{lem:aligned}
  Any set of $5$ points (or more) contains a monotonic path of length $3$.
\end{lemma}

\begin{proof}
  By the Erd\H{o}s-Szekeres theorem, given $\alpha$ and $\beta$, any
  sequence of distinct numbers of length at least
  $(\alpha-1)(\beta - 1) + 1$ contains a monotonically increasing
  subsequence of length $\alpha$ or a monotonically decreasing
  subsequence of length $\beta$. In two dimensions, one can first
  order the points according to $\le_x$, then consider the
  $y$-coordinates as the sequence of interest. The result follows by
  taking $\alpha = \beta = 3$.
\end{proof}

In the following subsections, we show how to wake up between $1$ and $6$ robots in a
square~$S$ of diameter~$1$ starting at one of its corner. For technical
reasons, we must
distinguish the case $n\le 5$ with makespan~$2$ (\cref{lem:S5}) and
the case $n\le 6$ with makespan~$3$ including the return to origin
(\cref{lem:S6}). Then, we give part of the proof of \cref{lem:T},
the remaining part being deferred to \cref{sec:S-T}.

\subsection{Proof of \cref{lem:S5}}


\LemmaSfive*


\begin{proof}
  Without loss of generality, suppose that the awake robot is located
  at a point $p_0$ at the left corner of the square. Thus,
  all the sleeping robots are on its right. We deal with a few cases separately:
  \begin{itemize}
    
  \item $n\le 3$: Any wake-up tree of depth~$2$ works, since
    the diameter of $S$ is~$1$.
		
  \item $n = 4$: By \cref{lem:aligned}, $S$ contains three robots
    whose positions define a monotonic path $(p_i,p_j,p_k)$ of length
    three (including possibly $p_0$). If we omit $p_j$, then by the
    previous case, any wake-up tree has makespan $2$. Pick a tree
    where one of the branch goes from $p_i$ to $p_k$. Then, we can
    insert $p_j$ on the way from $p_i$ to $p_k$ without impacting the
    makespan.
		
  \item $n = 5$: Either $p_0$ belongs to a monotonic path, or it does
    not. If it does, the initial robot wakes up the two corresponding
    robots in a single time unit, which gives us three awake robots,
    each of which wakes up one of the remaining robots in one time
    unit. If it does not, then among the sleeping robots, there is a
    unique robot $p^+$ whose $y$-coordinate is maximum, and a unique
    robot $p^-$ whose $y$-coordinate is minimum.
    Call $p_1,p_2,$ and $p_3$ the positions of the three remaining
    robots. Wlog, $p_1$ is leftmost among these points, and
    $p_2 \le_y p_3$. Now, if $p_2 \le_y p_1 \le_y p_3$, then
    $(p_0,p_1,p_2)$ or $(p_0,p_1,p_3)$ is a monotonic path from $p_0$.
    Thus, $p_1$ must also be topmost or bottommost. Wlog
    again, suppose that it is topmost and recall that it is also
    leftmost. If $p_1 \le_y p_0$, then again $(p_0,p_1,p_2)$ is a
    monotonic path from $p_0$, so $p_1 \ge_y p_0$,
    and since $(p_0,p_1,p^+)$ cannot be a monotonic path, we also
    have that $p_1 \ge_x p^+$. This implies that both
    $(p^+,p_1,p_2)$ and $(p^+,p_1,p_3)$ are monotonic paths. Based on
    these facts, the wake-up tree consists of first waking up $p^+$,
    resulting in two robots. One of them wakes up $p^-$, the other
    wakes up $p_1$. Then, one robot at $p_1$ wakes up $p_2$ and the other wakes up $p_3$. By monotonicity, the sleeping robots at $p_2$ and $p_3$ are both woken up within one time unit from $p^+$ (through $p_1$).
    \vspace{-\parskip}\vspace{-\parskip}
  \end{itemize}
\end{proof}

Note that \cref{lem:S5} is best possible in the sense that if $S$
contains exactly $6$ sleeping robots, then a makespan of $2$ may not
be achievable, which motivates the distinction between~\cref{lem:S5}
and~\cref{lem:S6}. The reader interested in this fact can have a look
at \cref{prop:13/6} and \cref{fig:counter_quarter}
in~\cref{sec:tightS5} (which is independent from our other results).

\subsection{Proof of \cref{lem:S6}}
\LemmaSsix*

\begin{proof}
  Again, suppose that the awake robot is
  located at the left corner (point $p_0$).
  We distinguish three cases:

  \Cyril{On dirait qu'il y a une hypothèse sur l'orientation gauche
    droite, haut bas de $S$. Il faut aussi virer le $r$ qui ne sert à
    rien.}

  \begin{itemize}

  \item \textbf{Case 1.} There exists a monotonic path
    $(p_0,p_i,p_j,p_k)$. By \cref{lem:monotonic}, we can have four
    awake robots located at $p_k$ in one time unit. Three of them can wake
    up the remaining three robots (separately) in one
    time unit, then each robot can move to $p_0$ in one time unit.

  \item \textbf{Case 2.} Case~1 does not hold but one or several
    monotonic paths $(p_0, p_i, p_j)$ exist. Among these, choose one
    that maximizes the $x$-coordinate of $p_j$, and among these (if
    several), choose one that minimizes the distance between the
    $y$-coordinate of $p_i$ and the $y$-coordinate of $p_0$. Wlog,
    assume that this path is monotonic in the NW direction. Since we
    are not in Case 1, no other robot may have a position $p_k$ that
    would cause $(p_0,p_k,p_i)$ or $(p_i,p_k,p_j)$ or $(p_i,p_j,p_k)$
    to be monotonic, nor $(p_0,p_i,p_k)$ to be monotonic with
    $p_k \ge_x p_j$ because $p_j$ is the rightmost such node. The
    forbidden regions are depicted in gray in the figure below (left).
    Note that this separates the authorized regions into an upper and
    a lower region (in white).

    
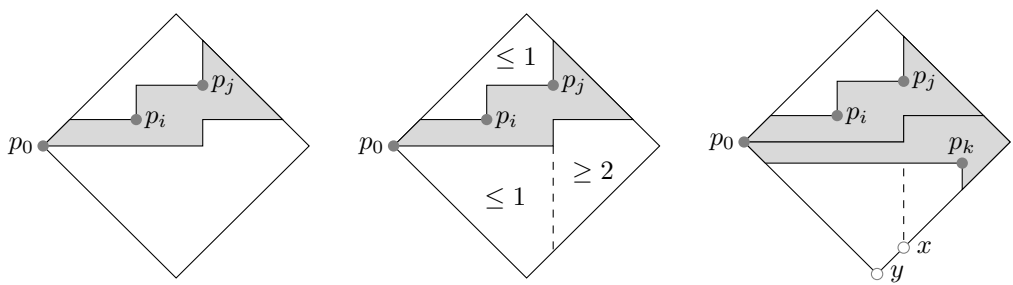
\begin{center}
    \begin{tikzpicture}[scale=3.5]
      \path (0,0) coordinate (A);
      \path (0.5,.5) coordinate (B);
      \path (1,0) coordinate (C);
      \path (0.5,-.5) coordinate (D);

      \path (0.35,0.1) coordinate (I);
      \path (0.35,0.23) coordinate (I2);
      \path (0.6,0.23) coordinate (J);
      \path (0.6,0.1) coordinate (J2);
      \path (0.6,0) coordinate (J3);
      \path (0.1,0.1) coordinate (X);
      \path (0.6,0.4) coordinate (Y);
      \path (0.9,0.1) coordinate (Z);
      \path (0,-0.55) coordinate (bidon);

      \draw (A)--(B)--(C)--(D)--(A);

      \draw[fill=lightgray!60] (A)--(X)--(I)--(I2)--(J)--(Y)--(Z)--(J2)--(J3)--cycle;

      \draw (A) node[inner sep=1.5pt,circle,fill=gray,label={[label distance=-2.5pt]180:$p_0$}] (A){};
      \draw (I) node[inner sep=1.5pt,circle,fill=gray,label={[label distance=-2.5pt]0:$p_i$}] (I){};
      \draw (J) node[inner sep=1.5pt,circle,fill=gray,label={[label distance=-2.5pt]0:$p_j$}] (J){};
    \end{tikzpicture}
    \hspace{.3cm}
    \begin{tikzpicture}[scale=3.5]
      \path (0,0) coordinate (A);
      \path (0.5,.5) coordinate (B);
      \path (1,0) coordinate (C);
      \path (0.5,-.5) coordinate (D);

      \path (0.35,0.1) coordinate (I);
      \path (0.35,0.23) coordinate (I2);
      \path (0.6,0.23) coordinate (J);
      \path (0.6,0.1) coordinate (J2);
      \path (0.6,0) coordinate (J3);
      \path (0.1,0.1) coordinate (X);
      \path (0.6,0.4) coordinate (Y);
      \path (0.9,0.1) coordinate (Z);
      \path (0,-0.55) coordinate (bidon);

      \draw (A)--(B)--(C)--(D)--(A);

      \draw[fill=lightgray!60] (A)--(X)--(I)--(I2)--(J)--(Y)--(Z)--(J2)--(J3)--cycle;

      \draw (A) node[inner sep=1.5pt,circle,fill=gray,label={[label distance=-2.5pt]180:$p_0$}] (A){};
      \draw (I) node[inner sep=1.5pt,circle,fill=gray,label={[label distance=-2.5pt]0:$p_i$}] (I){};
      \draw (J) node[inner sep=1.5pt,circle,fill=gray,label={[label distance=-2.5pt]0:$p_j$}] (J){};

      \draw[dashed] (J3)--(.6,-.4);

      \path[draw=none] (.46,.33) node {$\le 1$};
      \path[draw=none] (.42,-.2) node {$\le 1$};
      \path[draw=none] (.75,-.1) node {$\ge 2$};
    \end{tikzpicture}
    \hspace{.3cm}
    \begin{tikzpicture}[scale=3.5]
      \path (0,0) coordinate (A);
      \path (0.5,.5) coordinate (B);
      \path (1,0) coordinate (C);
      \path (0.5,-.5) coordinate (D);

      \path (0.35,0.1) coordinate (I);
      \path (0.6,0.23) coordinate (J);
      \path (0.82,-0.08) coordinate (K);
      \path (0.35,0.23) coordinate (I2);
      \path (0.6,0.1) coordinate (J2);
      \path (0.6,0) coordinate (J3);
      \path (0.6,-.1) coordinate (J4);
      \path (0.1,0.1) coordinate (X);
      \path (0.6,0.4) coordinate (Y);
      \path (0.9,0.1) coordinate (Z);
      \path (0.82,-0.18) coordinate (V);
      \path (0.08,-0.08) coordinate (W);
      \path (0.6,-0.4) coordinate (x);

      \draw (A)--(B)--(C)--(D)--(A);

      \draw[fill=lightgray!60] (A)--(X)--(I)--(I2)--(J)--(Y)--(Z)--(J2)--(J3)--cycle;
      \draw[fill=lightgray!60] (A)--(J3)--(J2)--(Z)--(C)--(V)--(K)--(W)--(A)--cycle;

      \draw (A) node[inner sep=1.5pt,circle,fill=gray,label={[label distance=-2.5pt]180:$p_0$}] (A){};
      \draw (I) node[inner sep=1.5pt,circle,fill=gray,label={[label distance=-2.5pt]0:$p_i$}] (I){};
      \draw (J) node[inner sep=1.5pt,circle,fill=gray,label={[label distance=-2.5pt]0:$p_j$}] (J){};
      \draw (K) node[inner sep=1.5pt,circle,fill=gray,label={[label distance=-2.5pt]$p_k$}] (K){};
      \draw (x) node[inner sep=1.5pt,circle,draw=gray,fill=white,label={[label distance=-1pt]0:$x$}] (x){};
      \draw (D) node[inner sep=1.5pt,circle,draw=gray,fill=white,label={[label distance=-1pt]0:$y$}] (y){};
      \draw[dashed] (J4)--(x);
    \end{tikzpicture}
  \end{center}
  
  Among the remaining points (whatever placed), either two points $p$
  and $p'$ exist such that $P_1 = (p_i,p,p',p_0)$ or
  $P_2 = (p_j,p,p',p_0)$ is a $2$-monotonic path, or no such pairs
  exist. If it exists, then the wake-up tree is as follows. The robot
  at $p_0$ wakes up $p_i$ and $p_j$ in one time unit. The robot at
  $p_i$ stays at $p_i$, so we have one robot at $p_i$ and two at
  $p_j$. Then, depending on whether $P_1$ or $P_2$ exists (if both
  exist, pick any), the corresponding robot wakes up $p$ and $p'$ and
  return with them at the origin. The other two wake up the last two
  robots (independently), and return with them at the origin. Overall,
  each robot has moved along a path that is at most $3$-monotonic.

  If neither $P_1$ nor $P_2$ exist, the strategy is different. First,
  observe that having two or more robots in the upper region would
  create a $2$-monotonic path from either $p_i$ or $p_j$ to $p_0$
  through these robots, thus the upper region contain at most one
  robot. By the same argument (from $p_j$ alone), the part of the
  lower region at the left of $p_j$ also contains at most one robot,
  so the situation is as depicted on the figure (middle). In
  particular, the lower right region is not empty. Let $p_k$ be the
  position of the rightmost robot in this region. Only one robot has
  such a $x$-coordinate, as otherwise, $p_j$ would have had a
  $2$-monotonic path through them towards $p_0$. For the same reason,
  none of the remaining robots are above $p_k$, and since $p_k$ is
  rightmost, no robots are on its right either, see the figure (right)
  for the remaining possible zones. Apart from $p_k$, the lower right
  region has between 1 and 3 robots; however, if they are at least
  $2$, then they must be aligned along a SW monotonic path from $p_k$,
  as otherwise (again), $p_j$ would have a $2$-monotonic path through
  them. Thus, a single monotonic path from $p_k$ can wake up all the
  robots in the lower right region. Finally, if there is a robot in
  the upper region, then it is possible to find a $2$-monotonic path
  from $p_k$ to $p_0$ going through $p_j$ and this robot. Now that all
  these facts are stated, the wake-up tree is as follows. Here, the
  robot at $p_0$ wakes up $p_k$ first. Then, using a $2$-monotonic
  path, one of the robot wakes up $p_j$ and the potential robot in the
  upper region, and return at $p_0$. The second robot at $p_j$ wakes
  up $p_i$ on its way to $p_0$. Meanwhile, the second robot at $p_k$
  wakes up all the robots in the lower right region using a monotonic
  path that finishes either at $x$ or $y$ (see the figure), depending
  on where the potential robot in the lower left region lies. This
  robot is woken up and all the robots return to $p_0$. All these
  movements are made along paths that are at most $3$-monotonic.

\item \textbf{Case 3.} No monotonic path of the form $(p_0,p_i,p_j)$
  exists. In this case, all the points below (resp. above) the
  $y$-coordinate of $p_0$ must form a $1$-monotonic path in the NE/SW
  direction (resp. NW/SE direction). In this case, we wake up the
  rightmost robot first. Then, one of the two robots wakes up the
  upper robots (if any) and the other wakes up the lower robots (if
  any). Finally, they all move to $p_0$. All these movements are made
  along paths that are at most $3$-monotonic.
  \vspace{-\parskip}\vspace{-\parskip}
  \end{itemize}
\end{proof}

\subsection{Proof of \cref{lem:T}}
\label{sec:lem-T}
\cref{lem:T} establishes that an arbitrary number of sleeping robots in
a triangle $T$ can be woken up within two time units. The approach is
inductive, namely, waking up a triangle often reduces to waking up
smaller nested triangles (containing strictly less robots), which explains
why the formulation of the lemma addresses several starting
configurations.

\LemmaT*

Without loss of generality, we assume that the triangle $T$ is
oriented as in \cref{fig:ABC-simple}, with vertices $ABC$ and hypotenuse
$[BC]$.
\begin{figure}[htbp!]
  \centering%
  \includesvg[width=0.4\textwidth]{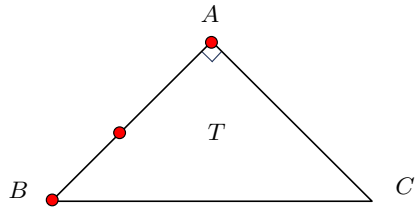}
  \caption{The triangle $T$ with vertices
    $B = (0,0)$, $C = (1,0)$ and $A = (1/2,1/2)$. In red, the possible starting points covered by the lemma.}
  \label{fig:ABC-simple}
\end{figure}
The goal is to show that all sleeping
robots in $T$ can be woken up in two time units, for each of the
possible starting configurations. Up to symmetry, these configurations are:

\begin{enumerate}
\item[]\textbf{Case A.} One awake robot is located in $A$.
  
\item[]\textbf{Case B.} One awake robot is located in $B$.
  
\item[]\textbf{Case C.} Two awake robots are located at
  a same point along segment $[AB]$.
\end{enumerate}

The strategy depend critically on how the robots are distributed
within the triangle, which gives rise to a number of subcases (14
overall). We agree that case-based proofs are not always satisfactory.
However, our proof is at least fully constructive (i.e., it yields an
actual algorithm). Furthermore, it is plausible that obtaining tight
bounds for this problem requires an unavoidable low-level scrutiny of
the instance. Indeed, many of the cases achieve the bound in a tight
way. We now proceed with the main three cases.

\subsubsection{Case A}
Here, the awake robot is located at the top of the triangle (point $A$, called $p_0$). This case does not rely on the same subdivisions as above. It uses a simpler recursion to smaller instances of Case A again, as shown in~\cref{fig:caseA}.

\begin{figure}[htbp!]
  \centering%
  \includesvg[width=0.5\textwidth]{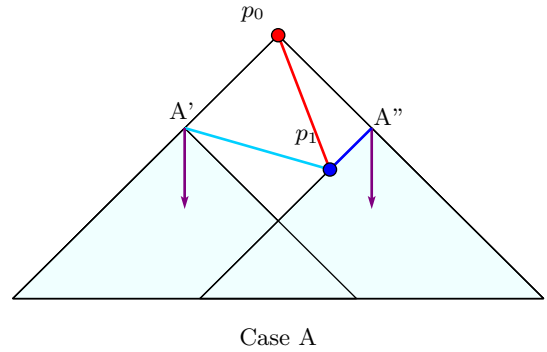}%
  \caption{In case A, the robot located at the top awakes the closest sleeping robot, then each of them applies recursively the algorithm in one of the two subtriangles. }
  \label{fig:caseA}
\end{figure}

Let $p_1$ be the closest sleeping robot from $p_0$. Let $d = |p_0p_1|$,
and let $S_A$ be the smallest square in $T$ that contains both $A$ and
$p_1$. Due to the $\ell_1$-norm, this square has diameter~$d$.
Furthermore, it is empty because $p_1$ is the closest point to $A$.
Let $A',A''$ be the points of $S_A$ intersecting $[AB]$ and $[AC]$,
respectively, and let $T'$ and $T''$ be the triangles defined
homothetically to $T$ with respect to $A'$ and $A''$. These
triangles have diameter $1-d$.

The wake-up strategy is as follows. The robot at $p_0$ wakes up the
robot at $p_1$. Then, one goes to $A'$ in order to wake up the robots
in $T'$, the other goes to $A''$ to wake up the robots in $T''$
(breaking ties arbitrarily if $T' \cap T''$ is not empty). Since $S_A$
has diameter $d$, any $2$ hops path has length at most $2d$, so the
two robots reach $A'$ and $A''$ before that time. Then, $T'$ and $T''$
are woken up in parallel (recursion of Case A), in at most $2(1-d)$ time unit, which gives a
total of $2d + 2 (1-d) = 2$ time units.

\subsubsection{Cases B}

The proof of Case B (and Case C) rely on a regular subdivision of $T$
into four smaller triangles of equal size. Call $D,E,F$ the middle
points of segments $[BC], [CA]$ and $[AB]$, respectively, and let
$T_A, T_B, T_C, T_0$ be the triangles $AFE$, $BDF$, $CED$, and $DEF$
(see \cref{fig:ABC}). Each of these triangles has diameter $1/2$.
Similarly, let $P_A$ and $P_B$ be the two parallelograms $AEDF$ and
$BDEF$. The diameter of $P_A$ is $1/2$, and the one of $P_B$ is $1$.

\begin{figure}[htbp!]
  \centering%
  \includesvg[width=0.7\textwidth]{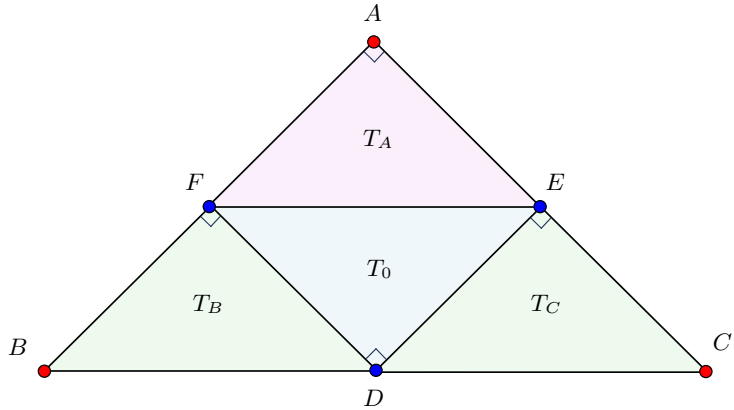}
  \caption{Canonical subdivision of the triangle $T$, with vertices
    $B = (0,0)$, $C = (1,0)$ and $A = (1/2,1/2)$.}
  \label{fig:ABC}
\end{figure}

Recall that in Case B, the awake robot start at point $B$, also referred to
as $p_0$. The case analysis depends on the distribution of nodes in the region
defined in the above subdivision, in particular the number of robots
in $P_B$ and $T_B$. A graphical summary of the subcases
is shown in \cref{fig:lemT-B}. The reader is encouraged to come back to
these pictures regularly.

\begin{figure}[htbp!]
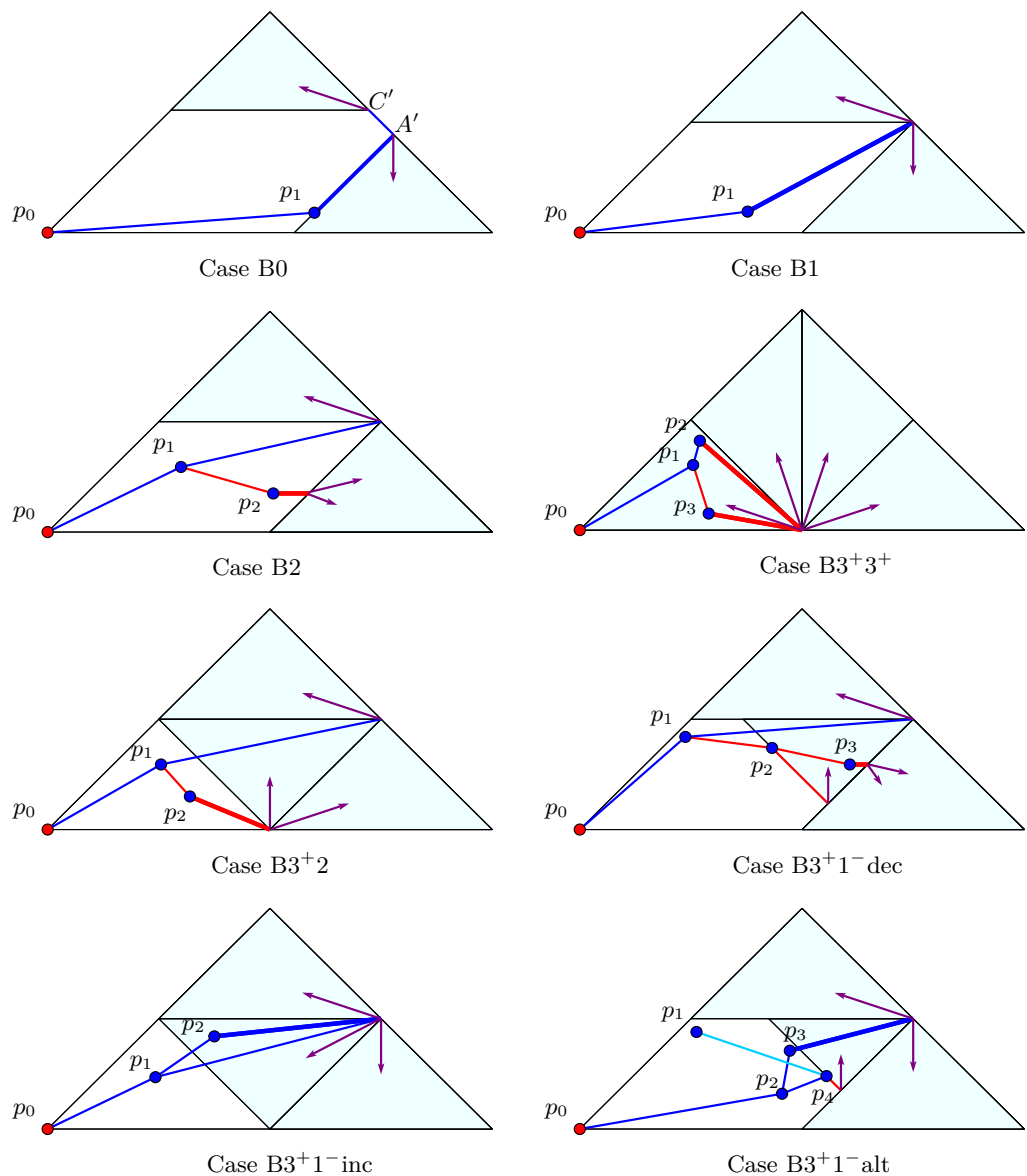

  \begin{center}
    \def\w{0.47}
    \def\j{2ex} 
    \begin{tabular}{@{}cc}
      \includesvg[width=\w\textwidth]{case_B0.svg}
      & \includesvg[width=\w\textwidth]{case_B1.svg}\\[\j]
      \includesvg[width=\w\textwidth]{case_B2.svg}
      & \includesvg[width=\w\textwidth]{case_B3+3+.svg}\\[\j]
      \includesvg[width=\w\textwidth]{case_B3+2.svg}
      & \includesvg[width=\w\textwidth]{case_B3+dec.svg}\\[\j]
      \includesvg[width=\w\textwidth]{case_B3+inc.svg}
      & \includesvg[width=\w\textwidth]{case_B3+alt.svg}\\[\j]      
    \end{tabular}
  \end{center}
  \caption{The 8 subcases of Case B (\cref{lem:T}). Regions in blue
    correspond to region where recursion occurs. Outgoing purple
    arrows indicates that the region will be woken up from the head
    location. A thick edge indicates that two awake robots follow the
    same path.}
	\label{fig:lemT-B}
\end{figure}


The first few cases depend on the number of sleeping robots in $P_B$. Namely, we apply $B_0$ if it is empty, $B_1$ if it contains one robot, and $B_2$ if it contains two robots. 

\begin{itemize}
	
\item \textbf{Case B0.} $P_B$ is empty. We increase the size of $P_B$
  homothetically, keeping one of its corners at $B$, until a point
  $p_1$ is found (see \cref{fig:lemT-B} - B0). The new parallelogram
  intersects with $[AC]$ in two points $C'$ and $A'$, where $C'$ is the
  highest (i.e. closest to $A$). This forms two smaller triangles
  which are homothetical to $ABC$. Because they result from
  intersecting a parallelogram, these triangles have the same size;
  namely, they have diameter $d=|AC'|=|A'C|$, which also implies that $|C'A'| = 1-2d$.

  The wake-up tree is as follows. The initial robot wakes up $p_1$.
  Depending on what side of the parallelogram $p_1$ lies on, both
  robots reach $C'$ or $A'$ using a path that is still monotonic from
  $p0$, so they arrive before one time unit. One of them then reaches
  the other point ($C'$ or $A'$) in time $1-2d$. Finally, each robot
  wakes up one of the two triangles (separately) in time $2d$,
  recursing into case A and B (respectively). Overall, the makespan is
  thus $1 + (1-2d)+2d=2$.

\item \textbf{Case B1.} $P_B$ contains one robot. The robot at $p_0$
  wakes up this robot, then both robots move to $E$ before one time
  unit, since the path $(p_0,p_1,E)$ is monotonic. Finally, one of
  them wakes up $T_A$ (recursing in Case B) and the other $T_C$
  (recursing in Case A). These triangles have half the size of $T$,
  thus the makespan is at most $1+2(1/2)=2$.

\item \textbf{Case B2.} $P_B$ contains two robots at $p_1$ and $p_2$.
  Wlog, assume $p_1 \le_x p_2$. If $(p_0,p_1,p_2)$ is monotonic, the
  strategy is the same as in Case $B_1$: the robots reach point $E$ in one
  time unit, then two of them wake up $T_A$ and $T_C$ independently.
  Otherwise, there exists a point $C^*$ of $[DE]$ such that $(p_1,p_2,C^*)$ is
  $1$-monotonic. In this case, the initial robot wakes up the robot in $p_1$, then
  moves to $E$ before one time unit and wakes up $T_A$ in
  $2(1/2)=1$ time unit. Meanwhile, the robot in $p_1$ wakes up the
  robot in $p_2$ and both move to $C^*$.

  Claim: The $2$-monotonic path $P=(p_0,p_1,p_2,C^*)$ has length at most one.

  Proof: Let $p_1=(x,y)$ and $C^*=(x',y')$. By 2-monotonicity, the
  length of the path is
  $|Bp_1| + |p_1C^*| = (x + y) + ((x'-x) + (y-y')) = 2y + x'-y'$. In
  terms of $y$-coordinate, the height of $T$ is $1/2$, thus the height
  of $P_B$ is $1/4$, and $y \le 1/4$. Moreover,
  because $C^*$ lies on $[DE]$, we have $y' = x' - 1/2$, so
  $2y + x'-y' \le 1/2 + x' - (x' - 1/2) = 1$.

  We thus have two robots located at $C^*$ before one time unit. These
  robot can wake up $T_C$ (of diameter $1/2$) in one
  time unit, by recursing in Case C.
\end{itemize}


The remaining cases address the configurations where $P_B$ contains at
least three robots. Here, we distinguish based on the number of robots
in the subtriangle $T_B$, namely whether $T_B$
contains three or more robots ($B3^+3^+$), two robots ($B3^+2$), or
only zero or one robot (B3$^{+}$1$^-$).

\begin{itemize}
\item \textbf{Case B3$^+$3$^+$.} $T_B$ contains at least three robots.
  We consider a slightly different subdivision of the part covered by
  triangles $T_A$ and $T_0$, dividing the corresponding area
  vertically into two equal triangles $ADF$ and $ADE$. Let $p_1, p_2$
  and $p_3$ be the first three points with respect to $\le_x$. The
  wake-up tree is as follows. The robot in $p_0$ wakes up the robot in
  $p_1$. Then, one of the two goes to $p_2$ and the other goes to
  $p_3$. Then, the four robots gather at $D$. Observe that all these
  paths from $p_0$ to $D$ are 2-monotonic, and $T_B$ has diameter $1/2$, thus the 
  robots arrive at $D$ before one time unit. Finally, each robot
  separately wakes up one of the triangles (of diameter $1/2$) in
  one time unit, by recursing in Case B.
	
\item \textbf{Case B3$^+$2.} $T_B$ contains two robots. Let $p_1$ and
  $p_2$ be the first two points with respect to $\le_x$. The wake-up
  tree is as follows. The robot at $p_0$ wakes up the robot at $p_1$.
  One of the robots goes directly to $E$, the other wakes up $p_2$ and
  the two resulting robots move to $D$. Observe that the path
  $(B,p_1,E)$ is $1$-monotonic, thus it has length $1$. The path
  $(B,p_1,p_2,D)$ is $2$-monotonic within a triangle of diameter
  $1/2$, so it has length $1$ as well. Finally, the robot in $E$ wakes
  up $T_A$ (recursing in Case B), one of the robots at $D$ wakes up
  $T_0$ (recursing in Case A), and the last robot wakes up $T_C$ (Case
  B again). All these triangles have diameter $1/2$, thus these
  recursive operations will take at most another time unit.
	
\item \textbf{Case B3$^+$1$^-$.} $T_B$ contains $0$ or $1$ robot. We have
  three subcases. For simplicity, we assume that $T_B$
  contains exactly one sleeping robot, thus $T_0$ contains
  two or more robots. The arguments are identical if $T_B$ is empty
  and all the robots of $P_B$ are in $T_0$.

  \Arnaud{Ces trois derniers cas sont à revoir. Beaucoup de non dit et
    l'ordre des hypothèses n'est pas encore très clair.}
  
  \begin{itemize}
		
  \item \textbf{Subcase B3$^+$1$^-$dec.} There exists a monotonic path
    $(p_1,p_2,p_3)$ in the SE direction that contains the point in
    $T_B$. Let $T_0'$ be the triangle resulting from shrinking $T_0$
    homothetically (keeping it anchored at $E$) until at least two of
    these points lie outside or along the side of $T_0'$
    (\cref{fig:lemT-B} - B3$^+$1$^-$dec). Call $D' \in [DE]$ the apex
    of $T_0'$. The wake-up tree is as follows. The robot at $p_0$
    wakes up the robot at $p_1$. One of them goes to $E$, the other
    wakes up $p_2$. Then, one of the robots in $p_2$ wakes up $p_3$
    and the other goes to $D'$. Finally, the two robots in $p_3$ move
    to any point $C^*X$ along $[DE]$ such that $(p_0,p_1,p_2,p_3,C^*)$
    is $2$-monotonic. From these locations, the algorithm recurses as
    follows: the robot in $E$ wakes up $T_A$ (Case B); the robot in
    $D'$ wakes up $T_0'$ (Case A); and the two robots in $C^*$ wake up
    $T_C$ (Case C). Since the path $(p_0,p_1,E)$ is $1$-monotonic, and
    $T_A$ has diameter $1/2$, $T_A$ will be woken up within another
    time unit, for a total of $2$ time units. Furthermore, both paths
    $(p_0,p_1,p_2,D')$ and $(p_0,p_1,p_2,p_3,C^*)$ are 2-monotonic
    within $P_B$. By the same argument as the claim in Case B2,
    both paths have length at most $1$. Thus, $T_C$ (of
    diameter $1/2$) and $T_0'$ (whose diameter is at most $1/2$) will
    also be woken up within two time units overall.
		
  \item \textbf{Subcase B3$^+$1$^-$inc.} There exists a monotonic path
    $P=(p_0, p_1, p_2)$ in the NE direction. In this case, $p_0$ wakes
    up $p_1$. One of them goes to $E$, the other wakes up $p_2$ and
    the two resulting robots go to $E$. From $E$, the three robots
    separately wake up $T_A$ (Case B), $T_0$ (Case B), and $T_B$ (Case A). By
    monotonicity, all of them arrive at $E$ before one time unit, and
    the three subtriangles have diameter $1/2$, thus the overall
    makespan is $2$.
		
  \item \textbf{Subcase B3$^+$1$^-$alt.} If we are
    neither in subcase B3$^+$1$^-$dec nor B3$^+$1$^-$inc, then the
    leftmost three points $p_1 \le_x p_2 \le_x p_3$ are such that
    $p_2 \le_y p_1$, $p_2 \le_y p_3$, and $p_3 \le_y p_1$. Thus
    $(p_0,p_2,p_3,E)$ is $1$-monotonic. The wake-up tree is as
    follows. The robot at $p_0$ wakes up $p_2$. Then, one of the two
    robots at $p_2$ wakes up $p_3$ and the two resulting robots move
    to $E$, where they wake up $T_A$ (recursing in Case B) and $T_C$
    (Case A). We are left with a robot at $p_2$. If $P_B$ contained
    exactly $3$ robots, then this robot wakes up $p_1$. Otherwise, let
    $T'_0$ be the triangle obtained by shrinking $T_0$ homothetically
    (keeping it anchored at $E$), until a new point $p_4$ lies on its
    side, and let $D'$ be the apex of $T_0'$. In this case, the path
    $(p_0,p_2,p_4,D')$ is 2-monotonic within parallelogram $P_B$, thus
    it has length at most $1$ (again, by the same claim as in Case
    B2). Thus, the robot at $p_2$ wakes up $p_4$. One of the resulting
    robot wakes up $p_1$, while the other move to $D'$ and wakes up
    $T'_0$ by recursing in Case A. Since $T'_0$ has diameter at most
    $1/2$, the overall makespan is again at most $2$.
  \end{itemize}
\end{itemize}

\subsection{Case C}
Due to space limitations, the proof of case C is deferred to
\cref{sec:S-T}. This proof is in the same spirit as the proof of
Case B and it also relies on the subvidision shown in \cref{fig:ABC}.


\section{Linear time algorithm}
\label{sec:linear}

In \cref{sec:wake-up-time}, we proved that the wake-up time of a unit $\ell_1$-disk can always be upper bounded by $5$ time units. The proof was constructive, but its time complexity is not linear. In this section, we prove that a linear time algorithm can asymptotically be achieved. More precisely, there exists a threshold $n_0$ such that if the number of sleeping robots $n$ is larger than $n_0$, then a wake-up tree of makespan less than $5$ can be computed in linear time in $n$ (\cref{th:time}). Thus, whenever $n < n_0$, one can use the constructive procedure from \cref{sec:wake-up-time}, then for larger values, one can use the linear time algorithm. Since $n_0$ is a constant, the computation time when $n< n_0$ is bounded by a constant, which implies that this combined strategy, overall, is a linear time algorithm. Due to space limitation, the content of this section is deferred to \cref{sec:linear-time}.

\section{Conclusion}
\label{sec:conclusion}

We have showed that in linear time one can produce a wake-up tree of
makespan at most five for robots in $L_1$. This wake-up ratio ``five''
is optimal: no strategy can guarantee less than five times the radius
under the $\ell_1$-norm. For $\ell_2$-norm, we have improved the best
known bound from $10.06$ to $7.07$. Some of our results are general
enough to apply to every norm.
We have also showed how to get in linear time a wake-up tree of
makespan no more than the wake-up ratio, for every norm.

Along the way, we have proposed a conjecture saying that, for every
norm $\eta$, the wake-up ratio is $1 + \Lambda$, where $\Lambda$ is
half the perimeter of the largest inscribed parallelogram of the unit
disk in $(\bbR^2,\eta)$. According to our results, the conjecture is
equivalent is saying that it is always quicker to wake up $n$ robots
than four. We have proved it for $\ell_1$ and $\elli$ norms.

As a first step towards this conjecture, it would be interesting to
determine the status of the $\ell_2$-norm whose wake-up ratio,
according to our conjecture, should be $1 + 2\sqrt{2} \approx
3.82$. Among $\ell_p$-norms, $\ell_2$ is the norm whose gap between
our upper and lower bounds on the wake-up ratio is the largest. In
spite of our efforts, we were unable to prove that, for instance, the
wort-case situation is whenever the points are all on a circle, and/or
equally distributed on the circle. One of the difficulty might be that
the longest branch, in a optimal (or near optional) wake-up tree, does
not necessarily form a convex set.

We have showed that the wake-up ratio for fixed $n$ asymptotically
decreases with $n$, i.e., $\gamma_n(\ell_2) < \gamma_4(\ell_2)$ for
large $n$ (more than $500$), but we were unable to show that this
inequality occurs for small $n$, say $n$ about $10$. Surprisingly,
experiments we have performed (see \cref{table:exp} in
\cref{sec:experiments}) show that one (at least) of the two following
likely statements must be wrong: (1) the wake-up ratio is reached for
points that are equally distributed on the unit circle; (2) for every
$n\ge 4$, $\gamma_{n+2}(\ell_2) < \gamma_n(\ell_2)$.

It might be difficult to find the exact bound of the wake-up ratio for
$\ell_2$, in the light of other constants in Computational
Geometry. This is notably the case for the \emph{stretch factor} of
the Delaunay triangulations, the maximum ratio between the distance
between any two points in the triangulation and their
$\ell_2$-distance. Despite a lot of efforts, current lower and upper
bounds for this stretch factor are $1.593$~\cite{BDLSV11} and
$1.998$~\cite{Xia13}. Gaps have been closed for $C$-Delaunay
triangulations (defined by some empty convex shape $C$), only for some
specific $C$, namely for
$C \in\set{ \mathrm{triangle,square,hexagon} }$,
see~\cite{Chew89,BGHP15,DPT21} respectively.

We summarize a list of further works:
\begin{itemize}

\item Calculate the wake-up ratio in $\ell_1$ or $\ell_2$ for a fixed
  number of $n>4$ of sleeping robots.

\item Prove or disprove that the wake-up ratio of the $\ell_2$-norm is
  $1+2\sqrt{2}$.

\item Prove or disprove that the wake-up ratio of the
  regular-hexagonal-norm\footnote{With this norm, it is easy to show
    that its unit disk (a regular hexagon) contains inscribed
    parallelograms of half-perimeter~$3$. This is clearly the largest
    possible length since~$3$ is also the half-perimeter of this disk
    (a hexagon).} is~$4$.

\item Prove or disprove \cref{conj} for $\ell_p$-norms.

\item Prove or disprove \cref{conj} for general norms.

\item Construct a linear time PTAS.

\item Extend the results to higher dimensions.

\end{itemize}



\bibliographystyle{my_alpha_doi}
\bibliography{a_tmp}


\newpage
\appendix


\section{End of the proof of \cref{lem:T}}
\label{sec:S-T}

\subsection{Case C}
We continue here with the third and last case, where two robots enter
the triangle $T$ through a point $C^*=p_0$ along its side $[AB]$. The
goal is to wake up $T$ in two time units, assuming that the diameter
of $T$ is normalized to $1$. As previously, the strategy depends on
the number of sleeping robots in certain subregions. It also depends
on whether the two robots are located along $[AF]$ or $[FB]$. Let
$P_B$ denote the parallelogram $BFED$ (i.e., the union of triangles
$T_B$ and $T_0$) and let $P_A$ denote the parallelogram $AEDF$ (union
of $T_A$ and $T_0$). Finally, let $P \in \set{P_A,P_B}$ be the
parallelogram that contains $p_0$. The main cases are as follows. If
$P$ contains no sleeping robots, we apply Case C0. Otherwise, the
strategy depends on the number of sleeping robots in the triangle
containing $p_0$. If it contains none, we apply Case C1; if it
contains exactly one, Case 2; and if it contains two or more,
Case C3. The cases are illustrated in \cref{fig:lemT-C}.




\begin{figure}[htbp!]
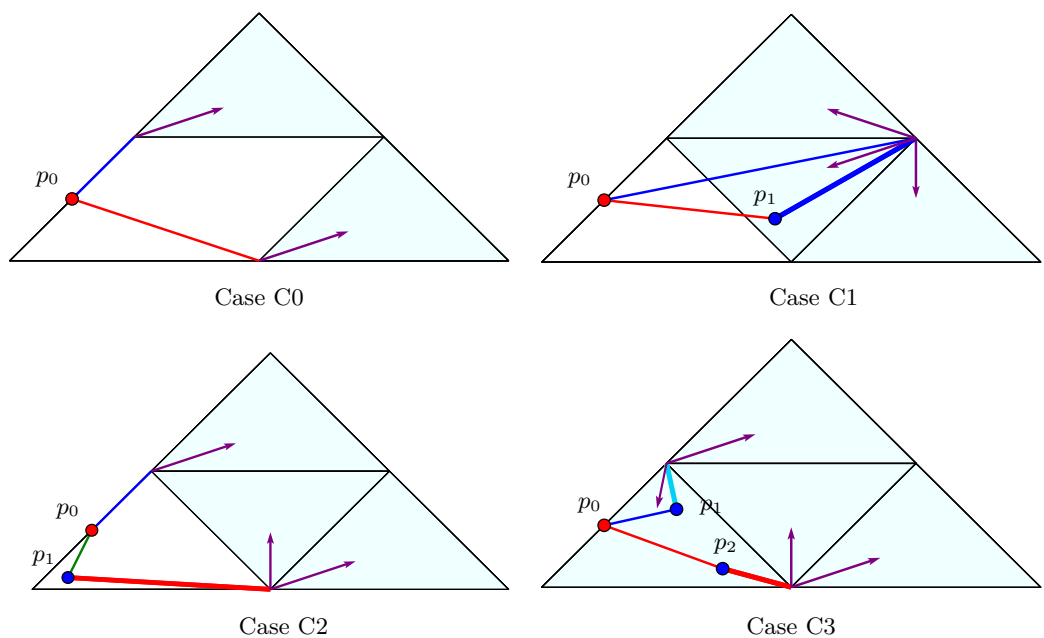

  \begin{center}
    \def\w{0.47}
    \def\j{2ex} 
    \begin{tabular}{@{}cc}
      \includesvg[width=\w\textwidth]{case_C0.svg}
      & \includesvg[width=\w\textwidth]{case_C1.svg}\\[\j]
      \includesvg[width=\w\textwidth]{case_C2.svg}
      & \includesvg[width=\w\textwidth]{case_C3.svg}
    \end{tabular}
  \end{center}
  \caption{The five subcases of Case C (\cref{lem:T}).}
	\label{fig:lemT-C}
\end{figure}


\begin{itemize}
	
\item \textbf{Case C0.} $P$ is empty. If $p_0 \in [FB]$, one robot
  goes to $D$ and wakes up $T_C$ (Case B), the other goes to $F$ and
  wakes up $T_A$ (Case B). If $p_0 \in [AF]$, one robot goes to $E$
  and wakes up $T_C$ (Case A), the other goes to $F$ and wakes up
  $T_B$ (Case A). It may happen that one of the subtriangles contains
  the same number of sleeping robots than $T$ itself, but since we
  recurse in Case A and Case B, the number of robots will
  inevitably decrease subsequently. The makespan is at most $2$.

\item \textbf{Cases C1.} $P$ is not empty and the triangle containing
  $p_0$ is empty. If $p_0 \in [FB]$, one of the robots goes directly to
  $E$, the other wakes up $p_1$. Then, the two resulting robots in
  $p_1$ go to $E$ as well. The path $(p_0,p_1,E)$ can always be
  realized through two monotonic parts, one in $T_B$ (of length at
  most $1/2$) and one in $T_0$ (same), thus it has length at most one.
  Then, the three robots wake up $T_A, T_0,$ and $T_C$ (each of
  diameter $1/2$) in one time unit by recursing in Case B, Case B, and
  Case A, respectively. If $p_0 \in [AF]$, one of the starting robots
  goes to $F$ (to wake up $T_B$ in Case A), the other wakes up $p_1$
  and move with it to $D$ before one time unit overall (by the same
  arguments). Finally, these two robots wake up $T_0$ (Case A) and
  $T_C$ (Case B) in another time unit.

\item \textbf{Cases C2.} $P$ is not empty and the triangle containing
  $p_0$ has exactly one sleeping robot, say at position $p_1$. If
  $p_0 \in [FB]$, one of the two robots goes directly to $F$. The
  other wakes up the robot at $p_1$ and the two resulting robots move
  to $D$. The path $(p_0,p_1,D)$ is at most $2$-monotonic within $T_B$
  of diameter $1/2$, thus these robots arrive at $D$ in at most one
  time unit. These two robots wake up $T_0$ (Case A) and $T_C$ (Case
  B) in another time unit. Similarly, the robot at $F$ wakes up $T_A$
  (Case B) in at most one time unit. If $p_0 \in [AF]$, one of the two
  robots goes directly to $E$, the other one wakes up $p_1$. One of
  them wakes up $T_C$ and the other goes to $F$ and wakes up $T_A$.
  The path $(p_0,p_1,F)$ is at most $2$-monotonic in a triangle of
  diameter $1/2$, thus all the robots are ready to wake up their
  assigned subtriangle before one time unit.

\item \textbf{Case C3.} $P$ is not empty and the triangle containing
  $p_0$ has at least two sleeping robots, say at positions $p_1$ and
  $p_2$. If $p_0 \in [FB]$, the two robots wake up (separately) $p_1$
  and $p_2$, which gives four awake robots. Two of them move to $F$,
  the two others move to $D$, arriving at these locations before one
  time unit (2-monotonic paths in a triangle of diameter $1/2$). From
  these locations, each of the four robots wakes up one of the four
  subtriangles. If $p_0 \in [AD]$, the strategy is the same, except
  that two of the four robots move to $F$ and the two others move to
  $E$, before waking up (separately) the four subtriangles.
\end{itemize}

\Cyril{J'ai viré toute la partie algo qui ne sert plus à rien, à mon
  humble avis. La preuve est constructive, cela suffit.}

\subsection{An illustrative scenario}

A more complex wake-up tree is shown on~\cref{fig:ex}, which involves
many different cases.\footnote{This construction was computed by an
  actual implementation of our algorithm.} From \cref{th:main}, we are
in the regime $n_0 \ge 11$, and thus we have to recruit to the densest
triangle first. Then, seven robots go back to the origin in order to
wake up the other seven triangles. In each triangle, \cref{lem:T}
applies. Along the induction, further subtriangles are considered. The
makespan may not be optimal, but it is lower than $5$
by~\cref{th:main}.

\begin{figure}[htb!]
  \centering%
  \includesvg[width=0.9\textwidth]{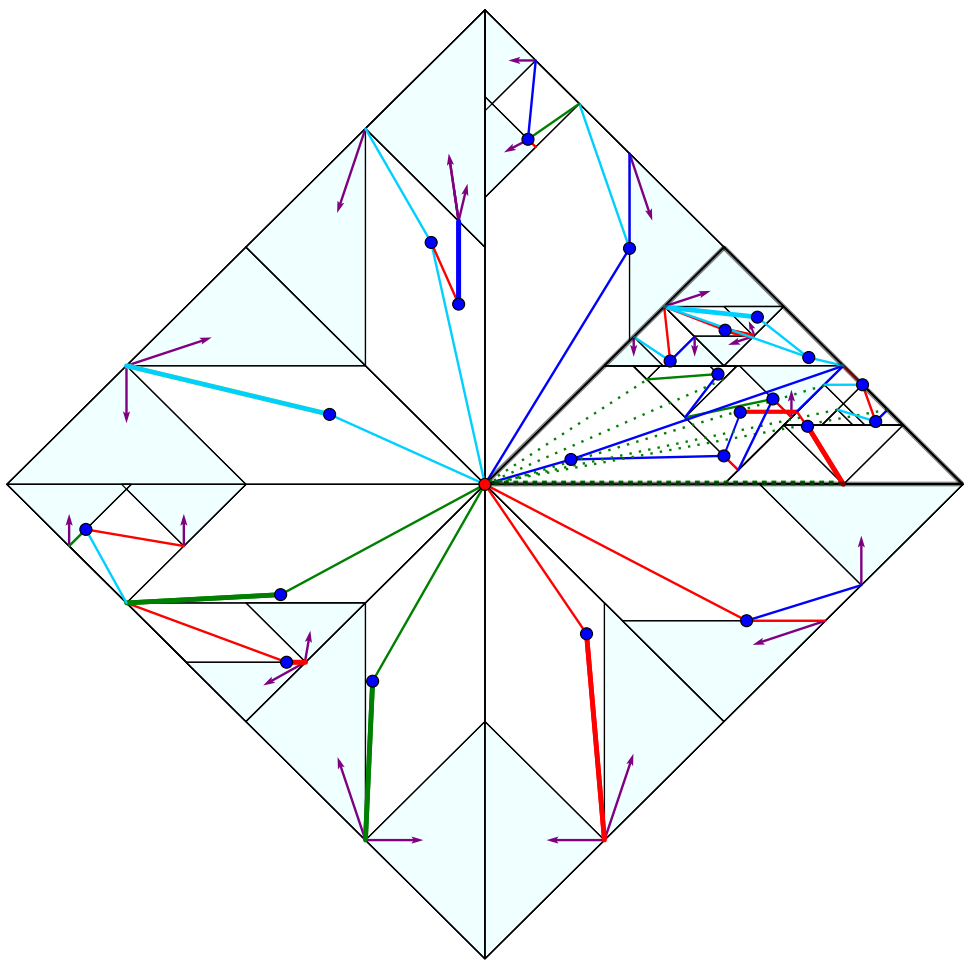}
  \caption{An illustration of our construction by applying the first
    steps of the inductive \cref{lem:T}. In these representation, blue
    triangles correspond to region where sleeping robots are not depicted. An
    arrow outgoing from a vertex $X$ in a triangle indicates that all
    the sleeping robots of this one will be woken up by the awake robots
    at position $X$. A thick edge indicates that two awake robots
    follow the same path. The bold triangle is the initial triangle
    where \cref{lem:T} is invoked. Dotted edges indicates the move
    of some robots to the origin. After the robots in the first triangle are
    awakened, seven awake robots come back to the origin to wake up
    the seven other triangles (using again \cref{lem:T}).}
  \label{fig:ex}
\end{figure}


\section{Detailed proof of \cref{th:time}}
\label{sec:linear-time}

The goal of this section is to prove \cref{th:time}. 

\TheoremLinearTime*

For this purpose, let us introduced two simple strategies:
\heapstr and \splitconestr. These strategies apply to
$(\bbR^2,\eta)$, for any norm $\eta$.

The positions of robots are represented by a point set
$P = \set{p_0, p_1, \dots, p_n}$ in the unit $\eta$-disk, where
$p_0 = (0,0)$ is the position of the awake robot, and $p_i$ is the
position of the $i$th sleeping robot, $i\in\range{1}{n}$.

\heapstr consists in building a minimum heap (binary) tree $H$ for
$\range{p_1}{p_n}$ where the key of $p_i$ is the distance from $p_0$
to $p_i$, i.e., $\eta(p_0-p_i)$. The wake-up tree rooted at $p_0$ is
then composed of $H$ itself, plus the edge connecting $p_0$ to the
root of $H$ (its top element), i.e., the closest point from
$p_0$. Using the well-known ``build-heap'' and ``heapify'' routines,
$H$ and thus the wake-up tree can be constructed in time $O(n)$.


\heapstr has the interesting property of constructing, in time
$O(n)$, a \emph{non-decreasing} wake-up tree for $P$: each robot is
always woken up by a robot which is closer to $p_0$ than it itself
is. In other words, for each edge $(p_i,p_j)$ of the tree, where $p_i$
is the parent of $p_j$, $\eta((p_0,p_i)) \le \eta((p_0,p_j))$. See
\cref{fig:heap}. As we will see, this strategy is efficient (it achieves
a low makespan) whenever $P$ is contained in a region of small width,
e.g., inside a parallelogram whose height is much smaller than its
length.

\begin{figure}[htbp!]
  \centering%
  \includesvg[width=0.6\textwidth]{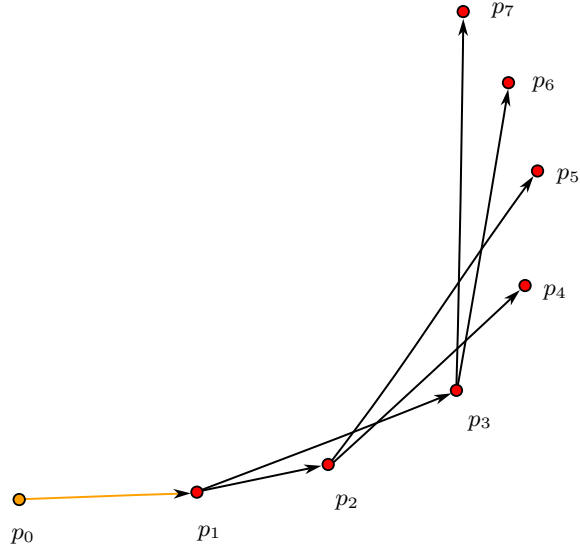}%
  \caption{The \heapstr applied to $p_0,p_1,\dots,p_7$, here in convex
    position. The resulting wake-up tree has the non-decreasing
    property (here w.r.t. $\ell_2$). Building the tree can be done in
    time $O(n)$, thus even faster than sorting or computing a convex hull.}
  \label{fig:heap}
\end{figure}

This property leads to a first application.

\begin{restatable}{proposition}{PropositionHeap}\label{prop:heap}
  If the points of $P$ are on a line, then an optimal wake-up tree for
  $P$ can be computed in $O(n)$.
\end{restatable}


\begin{proof}
  Let $L$ be the line of $(\bbR^2,\eta)$ containing $P$. By removing
  $p_0$ from $L$ split $(p_1,\dots,p_n)$ into two sets: $A$ and
  $B$. Let $p_a$ (resp. $p_b$) be the closest point of $A$ (resp. $B$)
  from $p_0$. And, let $p_{a'}$ (resp. $p_{b'}$) be the farthest point
  of $A$ (resp. $B$) from $p_0$.

  Observe that an optimal wake-up tree can always be transformed into a
  wake-up tree $T$ with same makespan whose first edge is
  $p_0 - p_u$ for some $u \in \set{a,b}$. This is because ``jumping''
  over $p_a$ or $p_b$ cannot improve the makespan. Then, from
  $p_u$, there must be a branch in $T$ that reach $p_{a'}$ and
  $p_{b'}$, leading to a branch (from $p_u$) of length at least
  $\max\set{\eta(p_u-p_{a'}), \eta(p_u-p_{b'})}$. In other words, the
  minimum makespan is at least
  \begin{equation}\label{eq:line}
    \min_{u\in\set{a,b}} \set{ \eta(p_0 - p_u) + M(u) },
    \qquad\mbox{where } M(u) = \max_{v\in \set{a',b'}} \set{\eta(p_u-p_v)}
  \end{equation}
  We construct a wake-up tree with such makespan using
  \heapstr as follows. Apply \heapstr to the subset $A$, with root
  $p_a$, giving a wake-up tree $T_A$. Since $T_A$ is
  non-decreasing, and $p_a$ is an endpoint of $A$, the makespan of $T_A$
  is precisely $\eta(p_a-p_{a'})$, $p_{a'}$ being the second endpoints
  of $A$. Similarly, applying \heapstr to $B$ gives a wake-up tree $T_B$
  with root $p_b$. This takes time $O(n)$.
\Arnaud{Vraiment pas fan de la notation $p_i - p_j$, qui se lit comme un moins.}
  From $T_A$ and $T_b$, we can construct a first wake-up tree $T'_A$
  for $P$ by connecting $T_A$ and $T_B$ with the edges $p_0 - p_a$ and
  $p_a - p_b$. This gives a valid wake-up tree with makespan
  $\eta(p_0-p_a) + \max\set{ \eta(p_a-p_{a'}), \eta(p_a-p_{b'}) }$,
  which is exactly $\eta(p_0-p_a) + M(a)$. Similarly, we can construct
  a second tree $T'_B$ for $P$ by connecting $T_A$ and $T_B$ with the
  edges $p_0- p_b$ and $p_b - p_a$. This gives a valid wake-up
  tree with makespan $\eta(p_0-p_a) + M(b)$. By taking the best of
  $T'_A$ and $T'_B$, we obtain
  in time $O(n)$ a wake-up tree of root $p_0$ of makespan
  $\min_{u\in\set{a,b}} \set{\eta(p_0-p_u) + M(u)}$, which is exactly
  the lower bound in Eq.\eqref{eq:line}.
\end{proof}

Note that \heapstr could replace efficiently the
\greedystr discussed in~\cite{SABM04,KLS05}: nearest sleeping robot
is awakened first\footnote{There are variants that depend on how
  conflicts between robot are resolved.}. The latter runs in
time $O(n^{2 - 2/(\ceil{d/2}+1) + \eps})$ for
$(\bbR^d,\ell_2)$, see~\cite{SABM04}, thus time
$O(n^{1+\eps})$ for $d\in \set{1,2}$. In fact, even for $d = 1$,
\greedystr requires
$\Omega(n\log{n})$ time for points on a line, by a simple reduction
from sorting $n$ numbers. It was proved that, for $d=1$, the
\greedystr leads to a $4$-approximation \cite[Th.~3]{SABM04} and that
the approximation ratio is at least $4-\eps$
(cf.~\cite[Th.~1]{SABM04}). \Arnaud{Thus, ?}

A second, and more important application is when the points of $P$
are in a \emph{cone}.
The \emph{unit circle}, w.r.t. the $\eta$-norm, is the boundary of the
unit disk. Let $\pi(\eta)$ be the \emph{half-circumference} of the
unit circle. So, the number $\pi \approx 3.14$ is nothing else than
$\pi(\ell_2)$. We known from  Go{\l}{\c{a}}b's Theorem that
$\pi(\eta) \in [3,4]$, both bounds being attained
for affinely regular hexagons and parallelogramms, respectively. (E.g.,
see~\cite[Th.4I-4K, pp.27]{Schaffer76}).
%
Given two points $A,B$ of the unit circle, denote by $\arc(A,B)$
the part of the circle that is traversed anti-clockwise
from $A$ to $B$ on the circle. The \emph{length} of $\arc(A,B)$ is
$|\arc(A,B)| \in [0,2\pi(\eta))$, measured in the $\eta$-norm.
Given a real $w \in [0,2\pi(\eta))$, and a point $A$ of the unit
circle, define $\cone(A,w)$ as the region of the plane composed of all
the points of the segments $[OX]$, where $X$ is the point of the unit
circle when going anti-clockwise from $A$ and such that
$|\arc(A,X)| = w$.
The value $w$ is called the \emph{arc-length} of $\cone(A,w)$. If
$\eta = \ell_2$, the arc-length of a cone corresponds to its angle.
See \cref{fig:cone-def}.

\begin{figure}[htbp!]
  \centering%
  \hspace{16ex}\includesvg[width=0.6\textwidth]{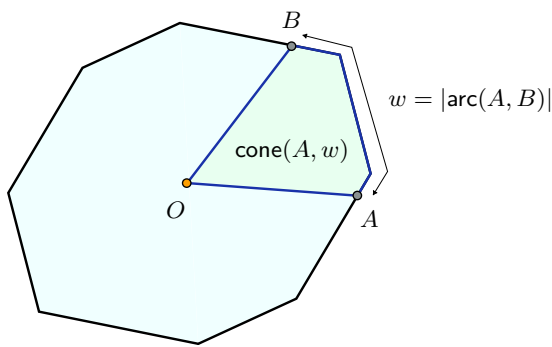}%
  \caption{The unit $\eta$-disk and $\eta$-circle for an arbitrary norm $\eta$,
    here given by a symmetric affinely octogon. The region in
    light-green is $\cone(A,w)$, with arc-length
    $w = |\arc(A,B)| \in [0,2\pi(\eta))$. Note that in general, two
    cones with same arc-length $w$, say $\cone(A,w)$ and
    $\cone(A',w)$, cannot be obtained from each other by a rotation
    around $O$.}
  \label{fig:cone-def}
\end{figure}

We have:

\begin{proposition}\label{prop:heap-cone}
  If $P$ is contained in a cone of arc-length $w$, then \heapstr
  constructs in time $O(n)$ a wake-up tree for $P$, rooted at the
  origin, with makespan at most $1 + w\floor{\log_2{n}}$.
\end{proposition}

\begin{proof}
  Assume $P \subset \cone(X,w)$ for some $X$ on the unit circle, and
  let $T$ be the wake-up tree produced by \heapstr. Let $\Gamma$
  be the arc of length $w$ from $X$, i.e., the intersection of
  $\cone(X,w)$ with the unit circle. Denote by
  $\lambda \Gamma = \set{ \lambda u \in\bbR^2 : u \in \Gamma }$ the
  arc $\Gamma$ scaled down by a factor $\lambda \in [0,1]$. Let
  $d(p) = \eta(p_0-p)$. By definition of the norm, every point $p$ on the arc of $q$
  satisfies $d(p) = d(q)$.

  Now, consider any edge $a-b$ of $T$, where $a$ is the parent of $b$. The
  homothetic arc of $\Gamma$ containing $a$ is $d(a) \Gamma$, whereas
  the arc containing $b$ is $d(b) \Gamma$. Let
  $a' = [Oa] \cap d(b) \Gamma$, the projection of $a$ along the
  segment $[Oa]$ on the arc of $b$. Similarly, let
  $b' = [Ob] \cap d(a) \Gamma$. See \cref{fig:cone-proof}.

  \begin{figure}[htbp!]
    \centering%
    \includesvg[width=0.8\textwidth]{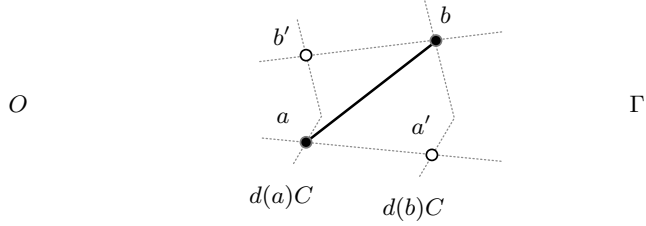}%
    \caption{Bounding the length of an edge $a-b$ of $T$ constructed
      by \heapstr for points in a cone (not represented) with apex
      $O$ and arc-length $|\Gamma|$.}
    \label{fig:cone-proof}
  \end{figure}
  
  W.l.o.g. assume $a \in d(a) \arc(X,b')$. The other case,
  $b' \in d(a) \arc(X,a)$ is similar. By the triangle inequality, we
  can bound the length of the edge $a-b$ by the length of the path
  $a-a'-b$. The latter is at most $|d(a') - d(a)| + |\arc(a',b)|$. We
  have $|\arc(a',b)| \le w$, and $d(a') = d(b)$ because $a'$ belongs
  to the arc of $b$. Moreover, since $T$ is non-decreasing,
  $d(a) \le d(b)$. Therefore, the length of $a-b$ is
  $|ab| \le d(b)-d(a) + w$.

  Consider any branch $(p_0,a_0,a_1, \cdots, a_h)$ of $T$. Its
  total length is bounded by:
  \[
    d(a_0) + \sum_{i=1}^h \pare{ d(a_i) - d(a_{i-1}) + w } ~=~ d(a_k) +
    wh ~\le~ 1 + w\floor{\log_2{n}} ~.
  \]
  Indeed, clearly, $d(a_k) \le 1$, and $h$ is the number of edges in
  the branch rooted at $a_0$. We conclude by the fact that $a_0$ and
  its descendants form a binary heap on $n$ elements, and thus of
  depth at most $\floor{\log_2{n}}$.
\end{proof}

In term of makespan, \heapstr is not optimal because it produces in
the wake-up tree branches that may zigzag in the cone, each turn
having possibly a cost of $w$ in the worst case. This can be corrected
using the \splitconestr. Roughly speaking, the strategy constructs
again a non-decreasing tree with the extra property that each subtree,
after the $i$ first steps, wakes up subcones whose arc-length becomes
exponentially smaller. This involves the golden ratio
$\varphi = (1+\sqrt{5})/2 \approx 1.61$.

\begin{proposition}\label{prop:split-cone}
  If $P$ is contained in a cone of arc-length $w$, then \splitconestr
  constructs in time $O(n\log{n})$ a wake-up tree for $P$, rooted at
  the origin, of makespan at most $1 + \varphi w$.
\end{proposition}

\begin{proof}
  Assume $P \subset \cone(X,w)$ for some $X$ on the unit circle, and
  let $\Gamma$ be the arc of length $w$ from $X$. Let
  $c = \varphi - 1 = 1/\varphi \approx 0.61$, where
  $\varphi = (1+\sqrt{5})/2$ is the golden ratio.

  The wake-up tree for $P$ is constructed as follows (\cref{fig:cone-split}).
  As for \heapstr, the first edge connect $p_0$
  to its closest (w.r.t. $\eta$-norm) sleeping robot at a
  position, say $a \in d(a) C$, where $d(a) = \eta(p_0-a)$. We then
  split the current $\cone(X,w)$ into two subcones $C$ and $C'$
  defined as follows: $C$ contains $a$ and its arc-length is $cw$,
  whereas $C'$ of arc-length $(1-c)w$ is the complementary cone of $C$
  in $\cone(X,w)$. Then, from $a$, the wake-up tree continues in
  parallel in $C$ and in $C'$, and connects $a$ with the closest point
  $b \in C$ and the closest point $b' \in C'$. The process continues
  recursively from $b$ within the subcone $C$, and from $b'$ within
  the subcone $C'$, according to the same rule of splitting: the
  current subcone $C$ or $C'$ being subdivided according to the ratio
  $c$ or $1-c$, the fraction $c$ containing the current point. We
  repeat the process until all the points have been spanned.

  \begin{figure}[htbp!]
    \centering%
    \includesvg[width=0.8\textwidth]{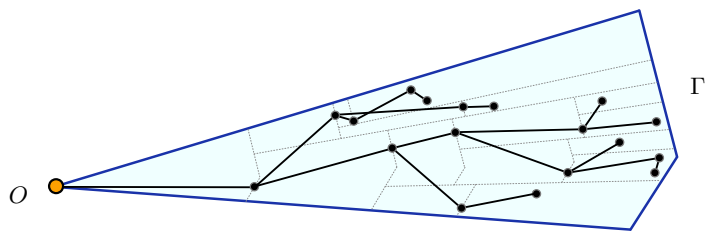}%
    \caption{Illustration of the \splitconestr, procuding
      non-decreasing wake-up trees. Each arc homothetic to $\Gamma$
      going thru a point $p_i$ is split into two sub-arcs, and defines
      two subcones: one of arc-length
      $(1/\varphi)|\Gamma| \approx 0.61 |\Gamma|$ (containing $p_i$)
      and its complementary of arc-length $\approx 0.39 |\Gamma|$.}
    \label{fig:cone-split}
  \end{figure}

  It is easy to check that the corresponding tree $T$ is
  non-decreasing and can be computed in time $O(n\log{n})$, as it
  requires to sort all the points according to the $\eta$-distance
  from $p_0$.

  It remains to analyze the makespan of $T$. Consider an edge $(a,b)$ of
  $T$, $a$ the parent of $b$. Using the triangle inequality (as we did
  in the proof of \cref{prop:heap-cone}), we can bound the length of
  $(a,b)$ by a contribution on the segment $[OX]$ and a contribution on
  the arc-lengths. Due to a telescopic sum ($T$ is non-decreasing),
  the total contribution on the segment $[OX]$ sums up to at
  most~$1$. For the contribution in the arc-length, we can proceed by
  induction. Assume that $a$ is in a subcone of arc-length $x\le w$,
  and denote by $f(x)$ the maximum arc-length contribution for a
  branch starting from $a$ to any leaf of $T$.

  Let us show that the arc-length contribution $f(x)$ fulfills the
  equation:
  \begin{equation}\label{eq:f(w)}
    f(x) ~\le~ \max\set{ cx + f(cx), x + f((1-c)x) } ~.
  \end{equation}
  Indeed, there are two cases.

  \begin{itemize}

  \item If $b \in C$, i.e., $b$ belongs to the same subcone of $a$,
    then the arc-length contribution for $(a,b)$ is at most $cx$ plus a
    contribution for any branch starting from $b$ to a leaf of $T$,
    which is by induction $f(cx)$ since all the descendents of $b$ in
    $T$ will be in $C$, a subcone of arc-length $cx$.
    
  \item If $b\in C'$, i.e., $b$ belongs to the complement subcone of
    $a$ with arc-length $(1-c)x$, then the arc-length for $(a,b)$ is at
    most $x$ plus the contribution for any branch starting from $b$
    which is $f((1-c)x)$ since $b$ belongs to $C'$.

  \end{itemize}

  By induction and by plugging $c = 1/\varphi$, it is not difficult to
  check that $f(x) < \varphi x$. Indeed, the first term of
  Eq.\eqref{eq:f(w)} gives,
  $cx + f(cx) < x(1/\varphi + 1) = \varphi x$. And, the second term
  gives, $x + f((1-c)x) < x(1 + (1-1/\varphi)) = \varphi x$. So, both
  terms are satisfied.

  The makespan of the \splitconestr is therefore at most
  $1 + f(w) < 1 + \varphi w$.
\end{proof}

The \splitconestr has an interesting corollary:

\begin{proposition}\label{prop:gamma_n}
  For any norm $\eta$, and every $n\in\bbN$,
  \[
    \gamma_n(\eta) ~<~ 3 +
    \frac{4\varphi\pi(\eta)}{\ceil{1+\sqrt{1+n}\,}} \quad
    \mbox{and}\quad \gme ~<~ 3 + \varphi\pi(\eta) ~\le~ 3 + 4\varphi
    ~.
  \]
\end{proposition}

\begin{proof}
  Let $k$ be the least integer such that $k(k-2) \ge n$. Since
  $k(k-2) = (k-1)^2 - 1$, this integer is
  $k = \ceil{1 + \sqrt{n+1}\,}$.

  In order to construct a low makespan wake-up tree, we shall use
  twice the \splitconestr as follows. We split the unit disk into $k$
  equal cones, so each with arc-length $w = 2\pi(\eta)/k$. At a first
  phase, we construct a wake-up tree in the densest cone by applying
  the \splitconestr.

  From \cref{prop:split-cone}, we obtain a wake-up tree of makespan at
  most $1 + \varphi w$. Because $k(k-2) \ge n$, the densest cone
  (among $k$) contains at least $k-2$ sleeping robots. When all of them
  are awakened, with have, with $p_0$, a total of $k-1$ awake robots
  (at least). For the second phase, we construct in parallel wake-up
  trees for the $k-1$ remaining cones thanks again to the
  \splitconestr.  Combining trees is possible, the number of awake
  robots contained in the wake-up tree\footnote{This number is
    precisely twice the number of leaves plus the number of vertices
    with one child.} during the first phase being at least the number
  of trees in the second phase. So, we can connect them.

  This leads to a wake-up tree of makespan less than
  $(1 + \varphi w) + 1 + (1+\varphi w) = 3 + 2\varphi w$. Plugging the
  values of $w$ and $k$, we get the claimed upper bound for
  $\gamma_n(\eta)$.

  To prove the second inequality, we construct a wake-up tree thanks
  to the following strategy: (1) wake-up any robot, and come back to
  the origin with two awake robots; and (2) wake-up in parallel each
  of the half-disk, that is a cone of arc-length $\pi(\eta)$, using
  the \splitconestr. The resulting makespan is less than
  $3 + \varphi\pi(\eta)$.

  The last inequality comes from the fact that $\pi(\eta) \le 4$, for
  every norm $\eta$.
\end{proof}

One can check for $\ell_2$-norm, by plugging $n = 528$ and
$\pi(\ell_2) = \pi$ in the equation of \cref{prop:gamma_n}, that
$\gamma_n(\ell_2) < 1+2\sqrt{2}$. Therefore, \cref{conj} --
\textit{It's always quicker to wake up $n$ robots than four} -- is
confirmed for $n \ge 528$.

The drawback of \splitconestr is that its construction does not take a
linear time. However, we can combined both strategies to get (almost)
the best of the both strategies. The resulting strategy, described in
the proof of \cref{prop:linear}, is called \linearstr.

\begin{proposition}\label{prop:linear}
  If $P$ is contained in a cone of arc-length $w$, the \linearstr
  constructs in time $O(n)$ a wake-up tree for $P$, rooted at the
  origin, with makespan at most $1 + \varphi w + o(w/n^{2/3})$.
\end{proposition}

\begin{proof}
  Assume $P \subset \cone(X,w)$ for some $X$ on the unit circle. Let
  $K = \ceil{n/\log_2{n}}$, and define $A \subset P$ composed of the
  $K$ closest points from $p_0$, breaking tie arbitrarily, and let
  $B = P\setminus A$. Using heap-sort, one can construct $A$ (and $B$)
  in time $O(n + K\log{n}) = O(n)$.

  We apply, the \splitconestr on $A$ that, from
  \cref{prop:split-cone}, produces in time $O(K\log{K}) = O(n)$ a
  wake-up tree $T_A$. We also subdivide $\cone(X,w)$ into $\ceil{w/s}$
  consecutive subcones, each of arc-length $s$, a number that will be
  fixed later. Let $C_i$ be the $i$th such subcones,
  $i \in \range{1}{\ceil{w/s}}$. We compute the sets
  $B_i \subset B \cap C_i$, the set of points of $B$ that fall into
  the subcone $C_i$ (breaking tie arbitrarily). Note that some $B_i$
  may be empty. We apply the \heapstr, independently for each
  $\set{p_0}\cup B_i$, set that is contained in $C_i$ (if not
  empty). This produces at most $\ceil{w/s}$ wake-up trees, one tree
  $T_{B_i}$ for each $\set{p_0}\cup B_i$. From \cref{prop:heap-cone},
  all the trees $T_{B_i}$ can be constructed in time
  $\sum_{i=1}^{\ceil{w/s}} O(|B_i|) = O(w/s) + O(\sum_i |B_i|) = O(w/s
  + n)$. We will require that $w/s = O(n)$.

  It remains to combine $T_A$ and $T_{B_i}$ trees. For that, we update
  each tree $T_{B_i}$ by removing its root $p_0$. Now, the arc-length
  $s$ is chosen large enough such that, if $B_i \neq\emptyset$, then
  at least one leaf of $T_A$ falls into $C_i$. Then, for each tree
  $T_{B_i}$, we connect its new root (the closest point of $B_i$ from
  $p_0$, since $p_0$ is not anymore in $T_{B_i}$) to any leaf of $T_A$
  that belongs to subcone $C_i$. This leads to the willing wake-up
  tree $T$ for $P$.

  From the analysis of the \splitconestr in the proof of
  \cref{prop:split-cone}, an edge of $T_A$ of depth $k$ leads to
  subcones of arc-length at most
  $t = \max_{i+j=k} (1/\varphi)^i (1-1/\varphi)^j w$. This is because
  at each edge, either the arc-length of the current cones is
  multiplied by a factor $(1/\varphi)$ or $(1-1/\varphi)$. Because
  $1/\varphi > 1-1/\varphi$, it follows that $t = w/\varphi^k$.  The
  depth of $T_A$, that spans $K$ points, is at least
  $\floor{\log_2{K}}$. So, the subcones of maximal depth and
  containing any leaf of $T_A$ are of arc-length most
  $w/\varphi^{\floor{\log_2{K}}} \le 2w/K^{\log_2{\varphi}}$. By
  choosing $s = 4w/K^{\log_2{\varphi}}$ (so twice larger), we ensure
  that the final subcones of maximal depth and containing any leaf of
  $T_A$ is contained is some subcones $C_i$ of arc-length $s$. We
  check also that $w/s < K^{\log_2{\varphi}} < K = O(n)$ as required.

  It remains to bound the makespan of $T$. Note that $T$ is
  non-decreasing. Therefore, the radius contribution of any branch is
  at most~$1$. For the arc-contribution, this is at most $\varphi w$
  for $T_A$, and then at most $s\floor{\log_2{n}}$ for $T_{B_i}$ (by
  \cref{prop:heap-cone} and \cref{prop:split-cone}). In total, the
  makespan is at most
  \begin{eqnarray}\label{eq:linear-mk}
    1 + s\floor{\log_2{n}} &\le& 1 + w\varphi + \frac{4w}{K^{\log_2{\varphi}}}
                                 \floor{\log_2{n}} ~\le~  1 + w\varphi
                                 + 4w \frac{\floor{\log_2{n}}}{\ceil{n/\log_2{n}}^{\log_2{\varphi}}}\\                       
                           &\le& 1 + w\varphi + o\pare{\frac{w}{n^{2/3}}}
  \end{eqnarray}
  noting that $\log_2{\varphi} \approx 0.69 > 2/3$.
\end{proof}

We are now ready to proof \cref{th:time}. Let us recall its
statement.

\TheoremLinearTime*

\begin{proof}
  Similarly to \cref{prop:gamma_n}, we can construct in time $O(n)$ a
  wake-up tree for $P$ with makespan at most $3 + c/\sqrt{n}$ for some
  constant $c$ large enough. Indeed, one can split the unit disk into
  $\sqrt{n}$ cones, each of arc-length $w = 2\pi(\eta)/\sqrt{n}$, and
  wake up the densest one. Then, using the \linearstr
  (\cref{prop:linear}) in this cone, containing $m \ge \sqrt{n}$
  points, we can wake up $m$ robots with a makespan (remember that
  $\pi(\eta) \le 4$):
  \begin{eqnarray*}
    1 + \varphi w + o(w/m^{2/3}) &=& 1 + \varphi w + o\pare{
                                     \frac{2\pi(\eta)}{\sqrt{n}}\bigg/ m^{2/3}} \\
                                 &=& 1 + \varphi w + o(n^{-5/6}) ~.
  \end{eqnarray*}
  Coming back to the origin, and repeating in parallel the \linearstr
  for all cones with some sleeping robots (at most $\sqrt{n}$ cones), we
  can complete the waking up. The time to build all these trees is
  $\sum_{i=1}^{\sqrt{n}} O(n_i) = O(\sqrt{n}\,) + O(\sum_i n_i) =
  O(n)$, where $n_i$ is the number of sleeping robots in the $i$th
  cone. The makespan of the construction is
  \begin{equation}\label{eq:n0}
    3 + 2\varphi w + o(n^{-5/6}) ~<~ 3 + \frac{26}{\sqrt{n}} + o(n^{-5/6})
    ~\le~ 3 + \frac{c}{\sqrt{n}}
  \end{equation}
  for a constant $c>26$ large enough (using the facts that
  $\pi(\eta) \le 4$ and that $8\varphi < 13$). Actually, $c$ can be
  precisely determined from Eq.\eqref{eq:linear-mk} in the proof of
  ~\cref{prop:linear}. The lowest order term in
  Eq.\eqref{eq:linear-mk} is $\le 4w$, for a single application of
  \linearstr. So, after two applications of the strategy, and plugging
  $w = 2\pi(\eta)/\sqrt{n} \le 8/\sqrt{n}$, we get a makespan of
  $3 + 26/\sqrt{n} + 8w \le 3 + (26+64)/\sqrt{n}$. Thus, $c \le 90$ is
  enough.

  Now, assume that $\tau > 3$ and $\tau \ge \gme$. Compute the least
  integer $n_0 \ge (c/(\tau-3))^2$. Note that $n_0$ is a fixed
  constant, independent of $n$. 

  \begin{itemize}

  \item If $|P| = n \ge n_0$, then we can apply the previous strategy
    providing a makespan that is less than
    $3 + c/\sqrt{n} \le 3 + c/\sqrt{n_0} \le \tau$ by Eq.\eqref{eq:n0}
    and by the choice of $n_0$.

  \item If $|P| = n < n_0$, then we can brute force for finding an
    optimal wake-up tree whose makespan is at most $\gme$ by
    definition of $\gme$. This is also at most $\tau$ by the choice of
    $\tau$. The number of wake-up trees we have to consider in a brute
    force algorithm is at most $n! \le n_0! = O(1)$, and checking the
    makespan of each of these trees costs $O(n) = O(1)$.

  \end{itemize}

  In both cases, we have constructed a wake-up tree in time $O(n)$ and
  with makespan $\le \tau$ as required. This completes the proof.
\end{proof}

We note that Eq.\eqref{eq:n0} in the proof of \cref{th:time} implies
an $(3+o(1))$-approximation running in time $O(n)$. A similar result
was already proved in~\cite[Th.~1]{ABGHM03}. However, our
construction, based on cones, gives a better second order term, namely
$O(1/\sqrt{n})$, whereas the $o(1)$ term given in the proof
of~\cite[Th.~1]{ABGHM03} is $\Omega(\log{n}/n^{1/4})$.


\section{Experiments}
\label{sec:experiments}

\def\unif{\mathrm{unif}}

We have done some experiments, and we have computed numerically, by a
brute force algorithm\footnote{Code available on demand to the
  authors. \Cyril{Astuce, on ne peut pas se citer soit même car c'est
    une soumission anonyme. On verra si accepté.}} the minimum
makespan for points that are equally distributed on the unit
circle. \cref{table:exp} shows the results for $\ell_2$-norm, but
results for other norms are available. We observe that, for this
distribution, the optimal makespan denoted by $\unif(n)$ are essentially
decreasing with $n$, for a given parity and $n\ge 4$, with some
exceptional cases.

\Cyril{Je voulais faire idem pour $\ell_1$, mais les points uniformes
  en $L1$ ne le sont pas en $L2$! (qui tombent en dehors de la boule
  $L1$.) Les résultats donnaient $>3$ dans le cas $n=3$ en $L1$,
  absurde. J'ai enlevé les résultats pour $n=0,1,2,3$. La raison est
  que pour $n=3$, on ne trouve pas $3$, mais $1+\sqrt{3} = 2.732$ ce
  qui est normal, mais plutôt troublant.}

\Cyril{La table ne tient pas horizontalement, j'ai essayé :((}

\begin{table}[!htbp]
  \def\f#1{\fbox{#1}}
  \def\b#1{\textbf{#1}}
  \begin{center}
    \begin{tabular}{ccc|ccc}
      & $n$ &&& $\unif(n)$ &\\
      \hline
      \hline
      &  4 &&& \b{3.828} &\\
      \hline
      &  5 &&&    3.351  &\\
      &  6 &&&    3.732  &\\
      &  7 &&&    3.431  &\\
      &  8 &&&    3.613  &\\
      &  9 &&&    3.416  &\\
      & 10 &&&    3.520  &\\
      & 11 &&&    3.383  &\\
      & 12 &&&    3.449  &\\
      & 13 &&&    3.349  &\\
      & 14 &&& \f{3.454} &\\
      & 15 &&&    3.318  &\\
      & 16 &&&    3.443  &\\
      & 17 &&& \f{3.331} &\\
      \hline
      \hline
    \end{tabular}
    \caption{Optimal makespan $\unif(n)$ (numerical approximation) for
      $n$ points equally distributed on the unit disk, with
      $\ell_2$-norm. Boxed values are exceptional cases such that
      $\unif(n) > \unif(n-2)$ and $n\ge 4$.}
    \label{table:exp}
  \end{center}
\end{table}

The exceptional cases imply that one the two following quite
reasonable statements is wrong: (1) the wake-up ratio is attained for
points that are equally distributed on the unit circle; (2) for every
$n\ge 4$, $\gamma_{n}(\ell_2) > \gamma_{n+2}(\ell_2)$.


\section{The Exact Value of
  \texorpdfstring{\boldmath{$\gamma_4(\eta)$}}{g4(eta)}}
\label{sec:gamma4}


\TheoremGammaFour*

\begin{proof}
  Consider a set of four points, $X = \set{A,B,C,D}$ (the sleeping
  robots), taken in the unit $\eta$-disk, and let $O = (0,0)$ be the
  origin, where the awake robot is placed.

  \myparagraph{Lower bound.} To show that
  $\gamma_4(\eta) \ge 1 + \lme$, assume that $X$ forms the largest
  parallelogram inscribed in the unit $\eta$-disk. Note that points
  are on the boundary of the unit disk. Any wake-up tree rooted at $O$
  and spanning $X\cup \set{O}$ must have a branch with at least three
  edges, say $e_1,e_2,e_3$. The first edge $e_1$ has
  length~$\eta(e_1) = 1$ since all points of $X$ are on the boundary
  of the unit disk. The next two edges $e_2,e_3$ must be taken among
  the $\binom{|X|}{2} = 6$ segments of $X$ (defined by any pair of
  points in $X$), namely
  $e_1,e_2 \in \set{s_1,s_2,s_3,\bar{s_1},\bar{s_2},\bar{s_3}}$, where
  $(s_1,s_2,\bar{s_1},\bar{s_2})$ correspond to the four consecutive
  sides of the boundary of $X$, and $s_3,\bar{s_3}$ correspond to the
  two diagonals of $X$. Because $e_2$ and $e_3$ must be consecutive
  segments of $X$ (say $e_2 = (A,B)$ and $e_3 = (B,C)$ for instance),
  we have $e_2 \in \set{s_i,\bar{s_i}}$ and
  $e_3\in \set{s_j,\bar{s_j}}$ for some $i \neq j \in
  \set{1,2,3}$. Clearly, $\eta(s_i) = \eta(\bar{s_i})$ and
  $\eta(s_i) \le \eta(s_3) = 2$. Because we want to lower bound
  $\eta(e_2)+\eta(e_3)$, we can assume that $i,j < 3$, i.e., $i=1$ and
  $j=2$ or the reverse. We conclude with the fact that
  $\eta(s_1) + \eta(s_2) = \lme$, and thus
  $\eta(e_1) + \eta(e_2) + \eta(e_3) \ge 1 + \lme$.

  \myparagraph{Upper bound.} It remains to prove
  $\gamma_4(\eta) \le 1 + \lme$.

  We will use the following facts.

  \begin{fact}\label{fact:1}
    If $C_1 \subset C_2$ are two convexes, then the perimeter of the
    boundary of $C_1$ is less than the perimeter of the boundary of
    $C_2$ (see \cite[Th.~4C p.~25]{Schaffer76} for instance).
  \end{fact}
  
  \begin{fact}\label{fact:2}
    Any quadrilateral contained in the unit $\eta$-disk has
    half-perimeter at most $\lme$.
  \end{fact}

  This is latter fact is a consequence of \cref{fact:1} and of the
  central symmetry of unit disk.


  It is well-known that in any set of five points contains four points
  in convex position (see~\cite{MS00a}). Note that in our setting,
  four points do not determine necessarily a quadrilateral since
  points are not necessarily in general position (and so some side may
  contain more than two points). Since we are concerned with the
  perimeter, for convenience, we will still call it a quadrilateral
  whereas we should speak about the four points on its convex hull.

  Let $Q$ be a subset of $X$ forming a quadrilateral, that is a convex
  having four points of its convex hull. There are two cases.
  

  \myparagraph{Case 1.} $O\notin Q$. In that case, we use the
  ``racquet'' strategy: $O$ goes to any point of $Q$ (at cost at most
  $1$); then in parallel, one robot turns clockwise and the other one
  anti-clockwise around the convex hull of $Q$ with an extra cost of
  half the perimeter of quadrilateral~$Q$. Overall the cost is at most
  $1 + \lme$ from \cref{fact:2}.

  \def\bA{\bar{A}}%
  \def\bB{\bar{B}}%
  \def\bC{\bar{C}}%


  \myparagraph{Case 2.} $O\in Q$. W.l.o.g. assume that $Q$ is $OABC$
  in this order. We have $D \notin Q$. Denote by $\bA,\bB,\bC$ be the
  opposite points of $A,B,C$ respectively, the symmetric points around
  $O$. There are two subcases.

  \myparagraph{Case 2a.} $C$ belongs to the convex hull of
  $A,B,C,\bA$. In that case the robot in $O$ goes to $A$ (at cost at
  most~$1$); then in parallel one robot goes to $D$ (with extra cost
  of~$2$), while the other goes to $B$ and then to $C$. The branch $(O,A,D)$ has
  length at most~$3 \le 1+\lme$ since $\lme \ge 2$. One can upper
  bound the length of the path $(A,B,C)$ by $\lme$. Indeed, we observe
  that, by translating the triangle $\bA\bB\bC$ by a vector $A-\bC$,
  one can form a parallelogram $ABC\bB$ that is contained in the unit
  disk (because it is included in the hexagon $ABC\bA\bB\bC$). It follows
  that the length of the path $(A,B,C)$, i.e., $\eta((A,B)) + \eta((B,C))$,
  is at most $\lme$. It follows that the length of the branch
  $(O,A,B,C)$ is at most $1 + \lme$.

  \myparagraph{Case 2b.} $C$ does not belong to the convex hull of
  $A,B,C,\bA$. It follows that $C$ is inside the triangle $AB\bA$. In
  that case, the robot in $O$ goes to $C$; then in parallel one robot
  goes to $D$ (with extra cost of $2$), while the other goes to $B$
  and then to $A$. The branch $(O,C,D)$ has length at
  most~$3 \le 1+\lme$ since $\lme \ge 2$. Remains to bound the length
  of the branch $(O,C,B,A)$. We first observe that $C$ is inside the
  subtriangle $OB\bA$. Indeed, $C$ cannot be inside the subtriangle
  $OAB$ since $Q$ is $OABC$ that is convex under Case~2
  hypothesis. The branch $(O,C,B,A)$ has a length that is bounded by the
  length of $(O,\bA,B,A)$. Indeed, by \cref{fact:1}, the triangle $OCB$
  has perimeter no more than the perimeter of the triangle
  $O\bA{}B$. We conclude with the fact that $AB\bA\bB$ is a
  parallelogram contained in the unit disk. Therefore, $(\bA,B,A)$ has
  length at most $\lme$. It follows that the length of the branch
  $(O,C,B,A)$ is at most the length of $(O,\bA,B,A)$ that is at most
  $\eta((\bA,O))$ plus the length of $(\bA,B,A)$, that is at most $1 + \lme$.
\end{proof}


\section{Proofs of~\cref{cor:ellp}}
\label{sec:proofs-cor:ellp}

\CorollaryEllp*

\begin{proof}
  The lower bound is a simple consequence of \cref{th:gamma4} and of
  the fact that $\Lambda(\ell_p) = 2^{1+\max(1/p,1-1/p)}$.

  For the upper bound, we use the inclusion of unit $\ell_p$-disk into
  $\elli$-disk, showing the well-known inequality
  $\ell_p(u) \le \elli(u) \cdot 2^{1/p}$, for every $u\in \bbR^2$.
  Moreover, by scaling and the inclusion of unit $\ell_p$-disk into
  $\ell_1$-disk, we have that
  $\ell_p(u) \le \ell_1(u) \cdot 2^{1-1/p}$. It follows that

  \begin{equation}\label{eq:ellp}
    \forall p \in[1,\infty],\quad \gamma(\ell_p) ~\le~ \min\pare{
      \gamma(\ell_1) \cdot 2^{1-1/p}, \gamma(\elli) \cdot 2^{1/p}} ~.
  \end{equation}

  We conclude with the fact that unit $\ell_1$-disk and unit
  $\elli$-disk have the same shape under rotation. So, we must have
  $\gamma(\ell_1) = \gamma(\elli)$, which is~$5$ by \cref{th:main}. The
  final, upper bound follows from Eq.\eqref{eq:ellp}.
\end{proof}

%
%


\section{Tightness of~\cref{lem:S5}}
\label{sec:tightS5}

\begin{proposition}\label{prop:13/6}
  There are six sleeping robots in a square of diameter~$1$ that
  requires a wake-up tree of makespan of at least $13/6$, if rooted at
  a corner.
\end{proposition}

\begin{proof}
  Let $ABCD$ be the vertices of the square of diameter~$1$, where
  $A = (0,0)$, $B = (1/2,1/2)$, $C = (1,0)$ and $D = (1/2,-1/2)$. The
  awake robot, the root, is placed at $p_0 = A$. Then, $p_2 = B$ and
  $p_1$ is at distance $\eps$ from $B$ for some $\eps \in [0,1/6]$.
  The four points $p_3,\dots,p_6$ are located on $[CD]$ such that they
  are pairwise at distance at least $2\eps$. This is possible if
  $\eps \le 1/6$. See \cref{fig:counter_quarter}. Observe that the
  distance between $[AB]$ and $[CD]$ is~$1$.

  \begin{figure}[htbp!]
    \centering%
    \includesvg{counter_quarter.svg}
    \caption{One optimal wake-up tree of makespan $13/6$ for $n = 6$
      sleeping robots in a square of diameter~$1$. \Cyril{Pour le
        dessin, j'ai viré les chemins $\ell_1$ qui me semblent
        rajouter de la confusion. En effet, pour le dessins
        \cref{fig:ex} par exemple, on n'utilise pas cette
        convention. Donc changer de convention à chaque dessin n'est
        pas cool. Et puis dire qu'on est dans le disque unité et de
        prendre des chemins qui n'y sont pas rajoute de la
        confusion.}}
    \label{fig:counter_quarter}
  \end{figure}
  Consider any wake-up tree $T$ for $p_1,\dots,p_6$ rooted at
  $p_0$. There are two cases:

  \begin{itemize}

  \item The first edge of $T$ starts in waking up some robots in
    $\set{p_1,p_2}$ before going to $\set{p_3,\dots,p_6}$. Then, after
    a time at least $1-\eps$, at most three robots are awake before
    going to $[CD]$ which contains four sleeping robots. Therefore, one
    of these sleeping robot will be wake up after an extra time of
    $1+2\eps$ since two robots of $[CD]$ are at distance at least
    $2\eps$. This gives a makespan for $T$ of at least
    $(1-\eps) + (1+2\eps) = 2 + \eps$.
    
  \item The first edge of $T$ starts in waking up some robots in
    $\set{p_3,\dots,p_6}$ before going to $\set{p_1,p_2}$. It follows
    that either $p_1$ or $p_2$ is wake up after time $2 +
    \eps$. Indeed, if $p_1$ and $p_2$ are woken up after a time
    $< 2 + \eps$, then $T$ must have the branches $(p_0,p_i,p_1)$
    and $(p_0,p_i,p_2)$. And, then one robot $p_k \in [CD]$,
    $k\neq i$, cannot be woken up in time better than~$3 > 2 + \eps$.

  \end{itemize}

  Overall, the makespan is at least $2+\eps$ that is $13/6$ if
  $\eps = 1/6$.
  %
  %
  %
  %
  %
\end{proof}

\end{document}